\documentclass[sigconf]{acmart}

\usepackage{enumitem}
\usepackage{fancyhdr}
\usepackage{bbm}
\usepackage{amsthm}
\newtheorem{definition}{Definition}
\usepackage{soul}
\usepackage{url}
\usepackage[utf8]{inputenc}
\usepackage{graphicx}
\usepackage{amsmath}
\usepackage{amsfonts}
\usepackage{booktabs}
\usepackage{algorithm}
\usepackage{algorithmic}
\usepackage[switch]{lineno}

\usepackage[n,advantage, operators, sets, adversary, landau, probability, notions, logic, ff, mm, primitives, events, complexity, oracles, asymptotics, keys]{cryptocode}

\AtBeginDocument{%
  }

\copyrightyear{2023} 
\acmYear{2023} 
\setcopyright{acmlicensed}\acmConference[CIKM '23]{Proceedings of the 32nd ACM International Conference on Information and Knowledge Management}{October 21--25, 2023}{Birmingham, United Kingdom}
\acmBooktitle{Proceedings of the 32nd ACM International Conference on Information and Knowledge Management (CIKM '23), October 21--25, 2023, Birmingham, United Kingdom}
\acmPrice{15.00}
\acmDOI{10.1145/3583780.3614998}
\acmISBN{979-8-4007-0124-5/23/10}
\begin{document}

\title{Online Efficient Secure Logistic Regression based on Function Secret Sharing}


\author{Jing Liu}
\authornote{Both authors contributed equally to this research.}
\affiliation{%
  \institution{East China Normal University}
  \city{Shanghai}
  \country{China}}
\affiliation{%
  \institution{Ant Group}
  \city{Hangzhou}
  \country{China}}
\email{jeanliu@stu.ecnu.edu.cn}

\author{Jamie Cui}
\authornotemark[1]
\affiliation{%
  \institution{Ant Group}
  \city{Hangzhou}
  \country{China}}
\affiliation{%
  \institution{East China Normal University}
  \city{Shanghai}
  \country{China}}
\email{shanzhu.cjm@antgroup.com}

\author{Cen Chen}
\authornote{Cen Chen is the corresponding author.}
\affiliation{%
  \institution{East China Normal University}
  \city{Shanghai}
  \country{China}}
\email{cenchen@dase.ecnu.edu.cn}








\newcommand{\emphsection}[1]{\smallskip\noindent\textbf{#1}}
\newcommand{\lj}[1]{\textcolor{red}{[#1]}}
\newcommand{\sz}[1]{\textcolor{blue}{[#1 -sz]}}
\newcommand{\ccc}[1]{\textcolor{green}{[#1 -cc]}}

\begin{abstract}
Logistic regression is an algorithm widely used for binary classification in various real-world applications such as fraud detection, medical diagnosis, and recommendation systems. However, training a logistic regression model with data from different parties raises privacy concerns. Secure Multi-Party Computation (MPC) is a cryptographic tool that allows multiple parties to train a logistic regression model jointly without compromising privacy. The efficiency of the online training phase becomes crucial when dealing with large-scale data in practice. 
In this paper, we propose an online efficient protocol for privacy-preserving logistic regression based on \textit{Function Secret Sharing (FSS)}. Our protocols are designed in the two non-colluding servers setting and assume the existence of a third-party dealer who only poses correlated randomness to the computing parties. During the online phase, two servers jointly train a logistic regression model on their private data by utilizing pre-generated correlated randomness. 
Furthermore, we propose accurate and MPC-friendly alternatives to the sigmoid function and encapsulate the logistic regression training process into a function secret sharing gate. The online communication overhead significantly decreases compared with the traditional secure logistic regression training based on secret sharing. We provide both theoretical and experimental analyses to demonstrate the efficiency and effectiveness of our method.
\end{abstract}

\ccsdesc[500]{Security and privacy~Domain-specific security and privacy architectures}

\ccsdesc[500]{Security and privacy}

\keywords{Secure multi-party computation, logistic regression, function secret sharing}


\maketitle

\section{Introduction}

Logistic regression is a machine learning model designed for binary classification tasks. 
Over the years, it has gained great popularity in various fields, including clinical prediction~\cite{nusinovici2020logistic}, 
financial prediction~\cite{kim2006logistic},
fraud detection~\cite{sahin2011detecting} and recommendation systems~\cite{cheng2016wide}, etc.
To build an effective and useful logistic regression model, a considerable amount of high-quality training data is typically necessary. The quality and quantity of the data directly impact the model's performance, with larger volumes and higher quality data leading to better model accuracy and more robust generalization. 
However, in reality, the data utilized for model training often belong to various individuals or organizations who seek to collaborate on training models but are hesitant to divulge their confidential raw data.
Moreover, the regulatory environment is continuously evolving and placing greater restrictions on access to sensitive data. As a result, the task of leveraging users' data for joint model training while maintaining the privacy of the data has become a significant challenge.

To address privacy concerns, various privacy-preserving computation techniques have been proposed, including Homomorphic Encryption (HE)~\cite{acar2018survey}, Differential Privacy (DP)~\cite{dwork2008differential}, and Secure Multi-Party Computation (MPC)~\cite{goldreich1998secure}. 
Existing MPC-based ML works~\cite{secureml, minionn, delphi,cryptflow2,Chameleon} use various techniques, such as secret sharing, garbled circuits, and also homomorphic encryption but they suffer from significant communication and computational overheads, which can lead to high online latency and communication costs. These limitations make MPC-based methods challenging to implement and deploy in practice. 
In the meantime, recent literature on MPC has introduced a new tool called \textit{Function Secret Sharing} (FSS) ~\cite{DBLP:conf/tcc/BoyleGI19}, which allows devising an almost ``optimal online communication cost'' for any functionality.
More specifically, FSS allows parties to generate generic input-independent correlated randomness (comparing with the Beaver's triples~\cite{Beaver91a}  in traditional MPC), which is later used by the online phase. 
In our paper, we assume that data owners do not have sufficient computational resources and, as a result, they outsource their data and computation to two high-performance cloud servers. To ensure data privacy and good utility, we leverage the recent advances of FSS in \textit{Secure Multi-Party Computation (MPC)}~\cite{goldreich1998secure,mohassel2018aby3} to make multiple data owners to jointly train a model while protecting the privacy of each party's data. The contributions of our work are summarized as follows:

\begin{itemize}
    \item We introduce a novel \textit{online efficient} secure logistic regression training protocol with the existence of a correlated randomness distributor (a.k.a. a dealer). Our protocol is secure in the semi-honest setting. 
    \item We propose two FSS-based approximations for the logistic function. Our first approximation comes from the Taylor series which is MPC-circuit-friendly since it only leverages multiplication and addition operations. Also, we leverage multiple interval containment (MIC) techniques for the second approximation to implement segmented evaluation for higher accuracy.
    \item Empirically, we compared our methods with traditional secure logistic regression training based on secret sharing implemented in various state-of-arts secure computation frameworks: SecretFlow~\cite{spu} and MP-SPDZ~\cite{keller2020mp}. Our results demonstrate that our approach achieves the fastest secure logistic regression training and the lowest communication overheads.
\end{itemize}
\section{Related Work}

The early instances of privacy-preserving machine learning based on MPC can be traced back to~\cite{DBLP:conf/psd/FienbergFSW06,DBLP:conf/icdm/SlavkovicNT07}. 
A variety of generic secure computation protocols have been developed to enable MPC-based machine learning.

\emphsection{Existing MPC techniques.} Secret sharing is a fundamental MPC tool that allows a group of participants to securely share and perform calculations on a secret. More specifically, a secret is divided into multiple shares by its owner, each of which is distributed to different participants. 
Several MPC frameworks~\cite{secureml,gazelle,delphi} leverage generic additively secret sharing to improve the security of distributed machine learning.  
In the meantime, the concept of \textit{Function Secret Sharing (FSS)} was originally introduced by ~\citeauthor{boyle2015function}~\cite{boyle2015function} for the purpose of private access to large distributed data while minimizing the overhead of communication. Initially, the FSS approach was extended to some simple functions, such as point functions, comparison functions, interval functions, and partial matching functions, and was used in MPC. To simplify previous FSS constructions and reduce the key size, the authors made some improvements and extensions~\cite{boyle2016function}, such as employing a tensoring operation for distributed point functions, leveraging pseudorandom generators, and reducing the intermediate keys. In addition, the researchers ~\cite{DBLP:conf/tcc/BoyleGI19} also focused on the preprocessing phase of secure computation, leveraging FSS schemes to compute different types of gates. More recently, they proposed a new approach~\cite{boyle2021function} which has considered smaller key sizes for commonly-used FSS gates. 

\emphsection{Logistic Regression based on MPC. } 
Privacy-preserving machine learning (PPML) has become a trending topic in recent years, with numerous research efforts focused~\cite{demmler2015aby, secureml,agrawal2019quotient,patra2021aby2,rathee2022secfloat,ariann}. Among them, secure training works can be categorized based on the computing parties and the security model involved. 
A typical scenario of PPML includes at least two computing parties, with a focus on the semi-honest threat model (since malicious protocols usually introduce too much overhead).
ABY~\cite{demmler2015aby} has introduced a new approach by leveraging techniques such as oblivious transfer, garbled circuits, and secret sharing to efficiently combine secure computation schemes based on Arithmetic sharing, Boolean sharing, and Yao's garbled circuits. 
SecureML~\cite{secureml} leveraged homomorphic encryption for preprocessing and enabled secure arithmetic operations on shared decimal numbers. They also introduced MPC-friendly alternatives to non-linear functions such as sigmoid and softmax, which were implemented using garbled circuits.
QUOTIENT~\cite{agrawal2019quotient},
ABY2.0~\cite{patra2021aby2} and
SecFloat~\cite{rathee2022secfloat} used secret sharing and oblivious transfer.
AriaNN~\cite{ariann} and Pika~\cite{wagh2022pika} used both secret sharing and function secret sharing.
Squirrel~\cite{lu2023squirrel} used secret sharing, homomorphic encryption, and oblivious transfer.
Over time, the number of parties involved in secure computation has gradually increased from two to many. Secure 3-party computation includes
ABY3~\cite{mohassel2018aby3},
Astra~\cite{chaudhari2019astra},
SecureNN~\cite{wagh2019securenn},
BLAZE~\cite{patra2020blaze},
Falcon~\cite{wagh2021falcon},
CryptGPU~\cite{tan2021cryptgpu},
Piranha~\cite{watson2022piranha} and
pMPL~\cite{song2022pmpl}. Among them, some work further considered more complex threat models, such as malicious servers or clients.
For secure 4-party computation, existing methods
FLASH~\cite{byali2019flash},
SWIFT~\cite{koti2021swift},
and Fantastic Four~\cite{dalskov2021fantastic} are all based on secret sharing.
\section{Preliminaries}
In this section, we aim to review the process of logistic regression and its implementation using MPC in the offline-online model. We will begin by introducing the basic notations.

\emphsection{Notations.} We denote a secret-shared value $x\in\mathbb{Z}_N$ as $\llbracket x\rrbracket$, where $N$ is a positive integer.
Also, we use $\llbracket x\rrbracket_0$ to denote $P_0$'s share of $x$, and $\llbracket x\rrbracket_1$ denote $P_1$'s share in $\ZZ_N$, where $x = \llbracket x\rrbracket_0 + \llbracket x\rrbracket_1\in\ZZ_N$. We also use $\gets$ to denote the assignment of variables, e.g., $x\gets 4$.

\subsection{The Process of Logistic Regression}
Given a training dataset with $m$ elements $\textbf{x}=\{x_1,x_2,\cdots,x_m\}$, each sample $x_i$ has an output label $y_i\in\bin$. Logistic regression employs weight coefficient $w$ and bias term $b$ to calculate a temporary output $g(x_i)=w\cdot x_i+b$ and then applies an activation function on $g(x_i)$ to map the output between zero and one. 
Most logistic regression algorithms utilize the logistic function, such as the sigmoid function, to perform non-linear mapping. The logistic function is represented by Equation~(\ref{initial sigma}).
\begin{equation}
    \sigma(x)=\frac{1}{1+e^{-x}}
    \label{initial sigma}
\end{equation}

Each $x_i$ may contain $n$ features. Therefore, we assume that input data is $\textbf{x}\in\mathbb{G}^{m\times n}$, weights are $\textbf{w}\in\mathbb{G}^{1\times n}$, bias are $\textbf{b}\in\mathbb{G}^{1\times m}$, and labels are $\textbf{y}\in\mathbb{G}^{m\times 1}$. 
 $\sigma$ is then used to map $\textbf{w}\textbf{x}^T+\textbf{b}$ into values between 0 and 1. 
For simplicity, the bias term \textbf{b} can be integrated into the weight coefficients \textbf{w}. Then, the predicted output label will be $\hat{\textbf{y}}^T=\sigma(\textbf{wx}^T)$.
The objective of logistic regression is to determine the optimal parameters \textbf{w} that minimize the discrepancy between the predicted label $\hat{\textbf{y}}$ and the true label \textbf{y}. To achieve this,
the cross-entropy function $C(\textbf{w})=(-\textbf{y}\log\hat{\textbf{y}}-(\mathbbm{1}-\textbf{y})\log(\mathbbm{1}-\hat{\textbf{y}}))/m$ is utilized as the cost function to measure the difference between the predicted and expected values. The calculation of $\textbf{w}$ is based on the optimization $\argmin_{\textbf{w}}C(\textbf{w})$.
The Stochastic Gradient Descent (SGD) is an iterative method for optimizing differentiable objective functions. The method updates the weights iteratively by computing the gradient of the loss function on the input data. Using SGD, we can approach a local minimum of the cost function. The SGD algorithm begins by randomly initializing \textbf{w}. The parameters are then updated along the gradient in the opposite direction. If we denote the learning rate as $\alpha$, the gradient for updating the weight in each iteration is computed by Equation ~(\ref{delta_w}).
\begin{equation}
    \Delta\textbf{w} = \alpha \cdot \frac{1}{m} \left(\sigma(\textbf{w}\textbf{x}^T)-\textbf{y}^T\right)\textbf{x}
    \label{delta_w}
\end{equation}

In each iteration, the weight is updated as $\textbf{w}:=\textbf{w}-\Delta\textbf{w}$ and is used in the next iteration. If the difference in accuracy between two adjacent epochs falls below a small threshold, it is assumed that \textbf{w} has converged to the minimum, and then the training process is terminated.
Researchers have found that when they applied MPC in logistic regression, the efficiency on both division and exponentiation was weaker~\cite{secureml}. To overcome this issue, they approximated division using Goldschmidt's~\cite{even2005parametric} or Newton's methods~\cite{markstein2004software} and the sigmoid function using Taylor series~\cite{aono2016scalable}. 
The direct use of the Taylor series to approximate the sigmoid function was found to be the most intuitive and effective method. To enhance the accuracy of the approximation and leverage the characteristics of the sigmoid function, a segmented function was utilized to transform the results of the Taylor approximation.

In this paper, we explore two distinct methods for approximating the sigmoid function. The first method involves the direct utilization of the Taylor series, while the second method employs a segmented function. The choice of the approximation method depends on the practical application scenarios. We provide an example of using the first-order Taylor series to illustrate the approximation of the sigmoid function. The approximation is represented by Equation ~(\ref{1-order taylor}).
\begin{equation}
    \sigma(x)=\frac{1}{2} + \frac{1}{4}x
    \label{1-order taylor}
\end{equation}
%

The Taylor series approximation method replaces exponentiation and division with multiplication and addition, thereby reducing computation complexity. However, a remaining problem is that the upper and lower bounds of the approximation are not zero and one as the sigmoid function. To maintain consistency with the general trend of the sigmoid function, when the input $x$ exceeds the upper threshold $\epsilon$, we directly set the output to 1, and when the input $x$ is below the lower threshold $-\epsilon$, the output is set to 0.
The segmented function for approximating the sigmoid function is defined by Equation ~(\ref{segmented function}).
\begin{equation}
\sigma(x)=
\begin{cases}
0, &x<-\epsilon \\
0.5+0.25x, & -\epsilon\leq x<\epsilon \\
1, &x\geq\epsilon
\end{cases}
\label{segmented function}
\end{equation}

\subsection{Our System Model}

In the literature, most MPC schemes~\cite{secureml, keller2020mp, SMPC} adopt MPC's preprocessing model, dividing the protocol execution into offline and online stages. The offline stage generates input-independent correlated randomnesses which are later used by the online stage. Interestingly, the offline operations could be either outsourced to a trusted third-party dealer, carried out by a trusted hardware (e.g. Intel SGX or ARM Trustzone)~\cite{noubir2016trusted}, or even emulated by a secure two-party protocol~\cite{yao1986generate} (e.g. Beaver Triple \cite{Beaver91a}). 

In this paper, we also adopted the pre-processing model and designed a novel secure logistic regression training protocol using function secret sharing technology, which improves online efficiency. 
 Figure~\ref{fig1} illustrates our initial framework for FSS-LR. 
\begin{figure}[!ht]
  \centering
  \includegraphics[width=0.9\linewidth]{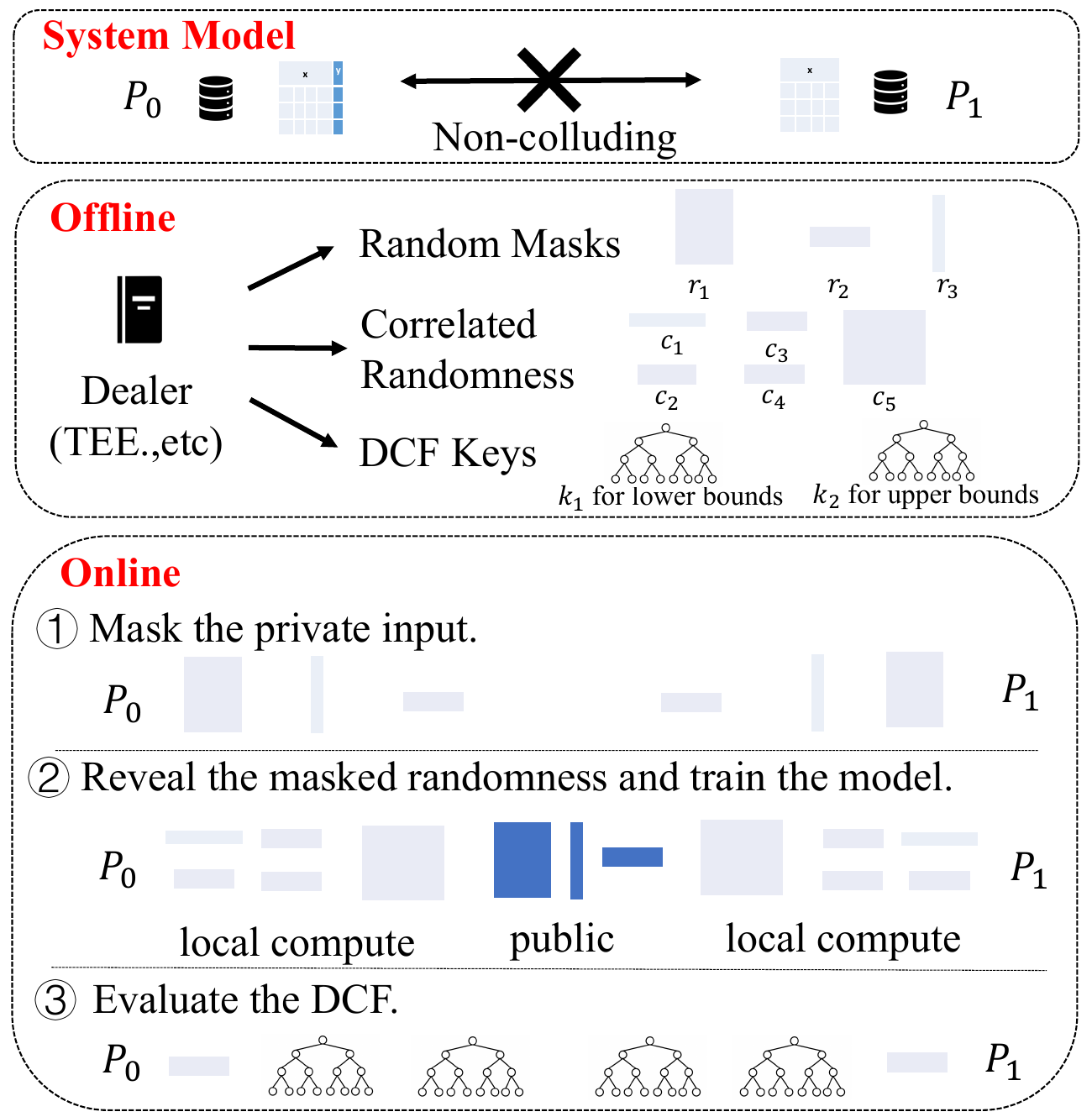}
  \caption{The Framework of FSS-LR.}
  \label{fig1}
\end{figure}
Our framework adheres to the standard simulation-based security definition \cite{DBLP:journals/eccc/Lindell17} in the semi-honest model. The security definition states that, given reasonable computational resources, an adversary cannot learn anything from the protocol except the function outputs.

\begin{definition}[Simulation-based Security Definition]
    We say a protocol $\pi$ is a secure instantiation of $f=(f_0, f_1)$ against semi-honest adversaries, if for all sufficiently large $\lambda\in\NN^*$, $x, y\in\bin^*$, there exists two PPT simulators $(\simulator_0, \simulator_1)$, such that the following holds,
    \begin{align*}
        \{(\simulator_0(1^\lambda, x, f_0(x, y)),  &f(x,y))\}_{x,y,\lambda}  \cindist\\
        \{(&\mathsf{view}^\pi_0(\lambda, x, y), \mathsf{output}^\pi(\lambda, x, y))\}_{x, y,\lambda}
        \\
        \{(\simulator_1(1^\lambda, y, f_1(x, y)),  &f(x,y))\}_{x,y,\lambda}  \cindist\\
        \{(&\mathsf{view}^\pi_1(\lambda, x, y), \mathsf{output}^\pi(\lambda, x, y))\}_{x, y,\lambda}.
    \end{align*}
    \label{def:simulation-security}
\end{definition}

\section{Our Proposed Protocol}
MPC or secret-sharing-based secure computation is a communication-heavy technique since it usually involves multiple rounds of interaction between each party for each operation. For instance in secret sharing logistic regression (SS-LR), to compute predicted labels and returned gradients, each party should open intermediate results twice and then use beaver triples to complete secure multiplication, resulting in significant communication costs and potential privacy leakage. To address this issue and improve online communication efficiency, we propose novel privacy-preserving logistic regression protocols that utilize function secret sharing and customized correlated randomness. 

\subsection{FSS-based Secure Computation} 
FSS is a rising secure multi-party computation technique~\cite{boyle2021function}, which aims to provide almost online-optimal secure computation. 
A two-server FSS scheme for a function family $\mathcal{F}$ splits a function $f\in\mathcal{F}:\bin^n$ into two additive shares $f_0,f_1$. Each share hides the original function $f$. For every input $x$, it always holds $f_0(x)+f_1(x)=f(x)$. Here, we assume both the input and output domains of $f$ are finite Abelian groups $\GG$, which we denote as $\GG^\text{in}, \GG^\text{out}$, respectively.
The share of the split function $f_i, i\in\bin$ is described by a compact key $k_i$ that is generated by the correlated randomness dealer.
Additionally, the dealer also sends out random values (e.g., $r_i$) to mask each party's private input.
Then, both parties reveal their masked inputs to each other and use previously generated keys to evaluate the split functions.
After the split functions have been evaluated, the results can be additively reconstructed or used as input for other downstream MPC applications.

\emphsection{Circuit and offset function family.}
In MPC, a basic operation (e.g., $+$, $\times $, $\xor$, etc.) can be represented by a gate. Each gate in the circuit is comprised of a pair of Abelian groups $(\GG^{\text{in}},\GG^{\text{out}})$, and a mapping $g\colon\GG^{\text{in}}\to\GG^{\text{out}}$. The combination of these gates forms a circuit, which represents a function family. The offset function is then used to adjust the output of the circuit to match the desired output. Using the FSS scheme, the function family $\hat{\mathcal{G}}$ of the offset function $\mathcal{G}$ is defined as the following,

\begin{align*}
&\left\{\begin{array}{l|l}
g^{\left[\mathrm{r}^{\mathsf{in }}, \mathrm{r}^{\mathsf{out }}\right]}: \mathbb{G}^{\mathsf{in }} \rightarrow \mathbb{G}^{\mathsf{out }} & \begin{array}{l}
g: \mathbb{G}^{\mathsf{in }} \rightarrow \mathbb{G}^{\mathsf{out }} \in \mathcal{G}, \\
\mathrm{r}^{\mathsf{in }} \in \mathbb{G}^{\mathsf{in }}, \mathrm{r}^{\mathsf{out }} \in \mathbb{G}^{\mathsf{out }}
\end{array}
\end{array}\right\}
\end{align*}
where we have $g^{\left[\mathrm{r}^{\mathsf{in }}, \mathrm{r}^{\mathsf{out }}\right]}(x):=g(x-\mathrm{r}^{\mathsf{in }})+\mathrm{r}^{\mathsf{out }}.$  $\mathrm{r}^{\mathsf{in }}, \mathrm{r}^{\mathsf{out }}$ are used to mask the input and output and help to achieve the offset function.

\subsection{FSS-based Logistic Regression}
Logistic regression is a mathematical model for classification. Traditional logistic regression training includes forward and backward propagation. During forward propagation, the predicted labels are computed and then passed through the sigmoid function to obtain values in the range of $[0,1]$. To simplify the computation of logistic regression, the sigmoid function can be approximated using the Taylor series, which allows us to express the logistic regression function as a function family of multiplication and addition. After that, we multiply the original data $\textbf{x}$ and the subtraction of predicted labels and the true labels to get the backward propagation gradient. In this way, the gradient $\Delta\textbf{w}$ can be computed as Equation ~(\ref{vector_delta_w}).
\begin{equation}
    \Delta\textbf{w}=\alpha\cdot\frac{1}{m}(\frac{1}{2}\cdot\mathbbm{1}+\frac{1}{4}\textbf{wx}^T-\textbf{y}^T)\textbf{x}
    \label{vector_delta_w}
\end{equation}

For higher accuracy, it is necessary to use segmented functions to constrain the approximation results within specific lower and upper bounds. Therefore, to calculate the backward gradient more precisely, the second approximation method using a segmented function is denoted as Equation ~(\ref{vector_spline}). 
\begin{equation} 
\begin{split}
    \Delta\textbf{w} =&-\alpha\cdot\frac{1}{m}\textbf{y}^T\textbf{x}\cdot\mathbf{1}\{\textbf{wx}^T_i<-\mathbf{\sigma}\} +\\
    & \alpha\cdot\frac{1}{m}(\frac{1}{2}\cdot\mathbbm{1}+\frac{1}{4}\textbf{wx}^T-\textbf{y}^T)\textbf{x}\cdot\mathbf{1}\{-\mathbf{\sigma}\leq\textbf{wx}^T_i<\mathbf{\sigma}\}+\\
    &\alpha\cdot\frac{1}{m}(\mathbbm{1}-\textbf{y}^T)\textbf{x}\cdot\mathbf{1}\{\textbf{wx}^T_i\geq\mathbf{\sigma}\} 
\end{split}
\label{vector_spline}
\end{equation}

After choosing two approximation methods, we further apply the FSS technique to design an offset function family gate specially for the logistic regression training process. Once the input is input into this gate, the gradient obtained from a single round of regression training will be output. This streamlined process requires only one round of reconstruction, reducing overall computation efforts and minimizing communication time. We denote the method using only the Taylor series as \textit{FSS-LR-V1}, while the method using a segmented function is denoted as \textit{FSS-LR-V2}.

In our protocol, the training data is secretly shared between two parties $P_0$ and $P_1$. We denote two parties' masked input data by $\llbracket \textbf{x}\rrbracket_0\in\mathbb{G}^{m\times n}$, $\llbracket\textbf{y}\rrbracket_0\in\mathbb{G}^{m\times 1}$ and $\llbracket\textbf{x}\rrbracket_1\in\mathbb{G}^{m\times n}$, $\llbracket\textbf{y}\rrbracket_1\in\mathbb{G}^{m\times 1}$. The coefficients \textbf{w}, initially set to be random, are secretly shared between two parties, denoted as $\llbracket \textbf{w}\rrbracket_0,\llbracket \textbf{w}\rrbracket_1\in\mathbb{G}^{1\times n}$. 
Among all the operations, the most expensive part comes from the two multiplication operations, one for forward predicting, and the other for backward gradient, respectively (1) $\mathbb{G}^{1\times n}\times\mathbb{G}^{n\times m}$, (2) $\mathbb{G}^{1\times m}\times\mathbb{G}^{m\times n}$. Our optimized method \emph{FSS-LR-V2} adds an additional implementation of the segmented function on top of the Taylor series.
Previous works~\cite{DBLP:conf/tcc/BoyleGI19,boyle2021function} provide several constructions of the \textit{Multiple Interval Containment (MIC)} gate for achieving interval partitioning. In our paper, we only introduce MIC with publicly known intervals.
\begin{definition}[Multiple Interval Containment (MIC)]
 The multiple interval containment gate $\mathcal{G}_\mathsf{MIC}$ could also be seen as a family of offset functions $g_{\mathsf{MIC}, n, m, P,Q}: \mathbb{U}_N\to \mathbb{U}^m_N$ for $m$ interval containments parameterized by input and output groups $\GG^\mathsf{in}=\mathbb{U}_N$ and $\GG^\mathsf{out} = \mathbb{U}_N^m$, respectively, and for $P=\left\{p_1, p_2, \ldots, p_m\right\}$ and $Q=\left\{q_1, q_2, \ldots, q_m\right\}$, given by
\begin{align*}
    \mathcal{G}_{\mathsf{MIC}}=\left\{g_{\mathsf{MIC}, n, m, P, Q}: \mathbb{U}_N\to\mathbb{U}_N^m\right\}_{0\leq p_i \leq q_i \leq N-1},\\
    g_{\mathsf{MIC}, n, m, P, Q}(x)=\left\{\mathbf{1}\left\{p_i \leq x \leq q_i\right\}\right\}_{1 \leq i \leq m},
\end{align*}
where we use $0\leq p_i\leq q_i\leq N$ to denote the $i$-th public interval.
\end{definition}
Existing MIC implementations leverage a crucial FSS building block \emph{Distributed Comparison Function (DCF)}. 

\begin{figure*}[t]
  \centering
  \includegraphics[width=0.9\linewidth]{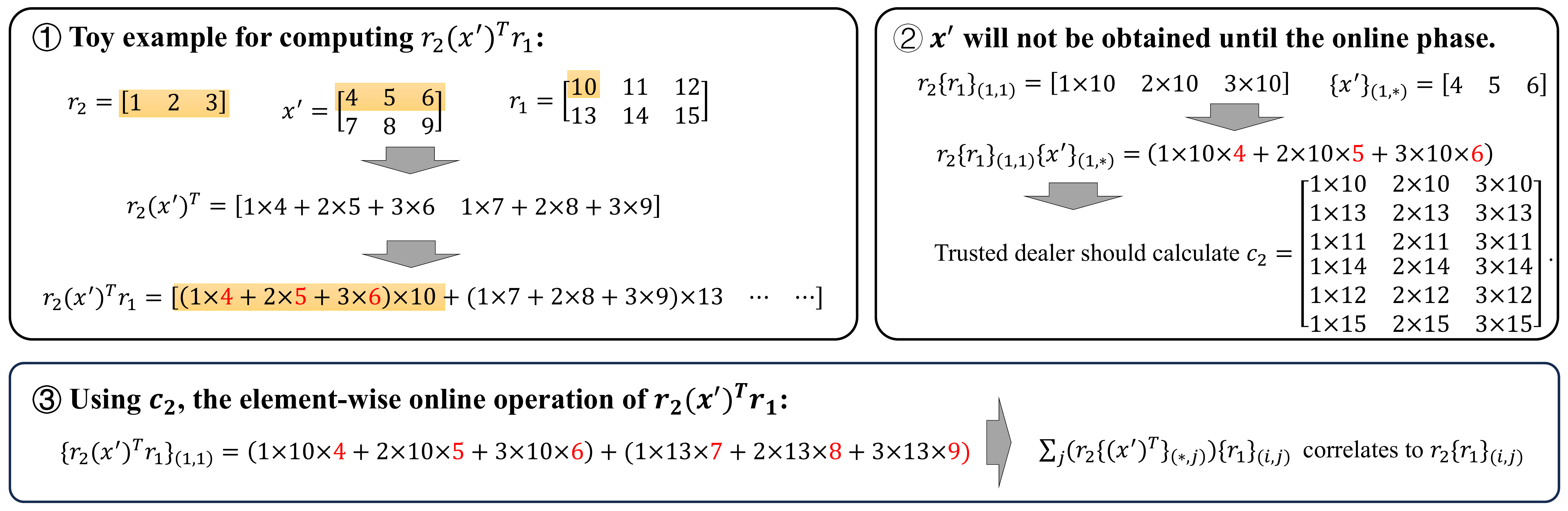}
  \caption{Toy example for correlated randomness $\textbf{c}_2$.}
  \label{fig2}
\end{figure*}
\begin{definition}[Distributed Comparison Function (DCF)]
    DCF is a special interval function $f_{\alpha, \beta}^{<}$, also referred to as a comparison function.  It outputs $\beta$ if $x<\alpha$ and 0 otherwise. Analogously, function $f_{\alpha, \beta}^{\leq}$ outputs $\beta$ if $x \leq \alpha$ and 0 otherwise.
\end{definition}
$\mathrm{DCF}_{n, \mathbb{G}}$ is used to count the number of invocations/evaluations as well as key size per evaluator $P_{b\in\bin}$. The FSS Gate for $\mathcal{G}_{\mathsf{MIC}}$ requires $2m$ invocations of ${\mathrm{DCF}}_{n, \mathbb{U}_N}$, and has a total key size of $mn$ bits plus the key size of $\mathrm{DCF}_{n, \mathbb{U}_N}$.
The details of how to use DCFs to build a MIC gate could be seen in \cite{DBLP:conf/tcc/BoyleGI19,boyle2021function}.
In our secure logistic regression training, $\mathcal{G}_{\mathsf{MIC}}$ is used to determine if the shares of the forward propagation $\textbf{wx}^T_i$ lie in the range $(-\infty,-\epsilon],(-\epsilon,\epsilon],$ or $(\epsilon,+\infty)$.  
Inside the MIC gate, the distributed comparison function $f_{\sigma,\mathbbm{1}}^{<}$ is first used to find whether the element in $\textbf{wx}^T$ is less than $\sigma$ or not. If not, the activation value of $\textbf{wx}^T$ should be $\mathbbm{1}$. Otherwise, it should be put into a further evaluation function $f_{-\sigma,\mathbbm{0}}^{<}$. Therefore, we can use two DCFs to partition three intervals. Then, multiplying the output of the $\mathcal{G}_{\mathsf{MIC}}$ gate by the segmented approximation result can get the final approximation of the sigmoid function output. According to the sigmoid function's different approximating modes, the weight's backpropagation result is obtained differently. Therefore, the overall logistic regression function family circuit consists of gates of addition, multiplication, and comparison. We assume the logistic regression gate $\mathcal{G}_{LR}$ as 
\begin{align*}
&\left\{\begin{array}{l|l}
g^{\left[\textbf{r}^{\mathsf{in }}_1, \textbf{r}^{\mathsf{in }}_2,\textbf{r}^{\mathsf{in }}_3 , \textbf{r}^{\mathsf{out }}\right]}_{LR}: \mathbb{G}^{\mathsf{in }} \rightarrow \mathbb{G}^{\mathsf{out }} & \begin{array}{l}
g: \mathbb{G}^{\mathsf{in }} \rightarrow \mathbb{G}^{\mathsf{out }}, \\
\left\{\textbf{r}^{\mathsf{in }}_1, \textbf{r}^{\mathsf{in }}_2, \textbf{r}^{\mathsf{in }}_3 \right\}\in \mathbb{G}^{\mathsf{in }}, \\
\textbf{r}^{\mathsf{out }} \in \mathbb{G}^{\mathsf{out }}\\
\end{array}
\end{array}\right\}
\end{align*}
where
$g^{\left[\textbf{r}^{\mathsf{in }}_1,\textbf{r}^{\mathsf{in }}_2,\textbf{r}^{\mathsf{in }}_3 \textbf{r}^{\mathsf{out }}\right]}(\textbf{x},\textbf{w},\textbf{y}):=
g\left(\textbf{x}-\textbf{r}^{\mathsf{in }}_1,\textbf{w}-\textbf{r}^{\mathsf{in }}_2,\textbf{y}-\textbf{r}^{\mathsf{in }}_3\right)+\textbf{r}^{\mathsf{out }}$.

The online computation determines what kind of correlated randomness needs to be produced in the offline phase. During the offline phase, we let the trusted dealer first use the pseudo-random generator (PRG) to sample three independently uniform random matrices $\textbf{r}_1\in\GG^{m\times n}, \textbf{r}_2\in\GG^{1\times n}, \textbf{r}_3\in\GG^{m\times 1}$, which are used to mask the \textbf{x}, \textbf{w}, \textbf{y} respectively. In the later online phase, we will use them to help reveal $\textbf{x}^\prime=\textbf{x}-\textbf{r}_1, \textbf{w}^\prime=\textbf{w}-\textbf{r}_2, \textbf{y}^\prime=\textbf{y}-\textbf{r}_3$.
Apart from random matrices, the dealer also needs to generate correlated randomness, $\textbf{c}_1=\textbf{r}_2\textbf{r}_1^T\in\GG^{1\times m}$, $\textbf{c}_2=\textbf{r}_2(\textbf{x}^\prime)^T\textbf{r}_1\in\GG^{n\times m\times n}$, $\textbf{c}_3=\textbf{r}_2\textbf{r}_1^T\textbf{r}_1\in\GG^{1\times n}$, $\textbf{c}_4=\textbf{r}_3^T\textbf{r}_1\in\GG^{1\times n}$, $\textbf{c}_5=\textbf{r}_1^T\textbf{r}_1\in\GG^{n\times n}$.

Among all the correlated randomness, we take the most complex correlated randomness $\textbf{r}_2(\textbf{x}^\prime)^T\textbf{r}_1\in\GG^{1\times n}$ as an example to illustrate. Let $\Delta\sample\GG^{1\times n}$ be the result, for $i\in\set{1,...,n}$, we have the value of $\textbf{r}_2(\textbf{x}^{\prime})\textbf{r}_1$ at the $i$-th position as follows:
\begin{align*}
    \Delta_{1, i} = & \set{\textbf{r}_2(\textbf{x}^\prime)^T\in\GG^{1\times m}} \cdot\set{\textbf{r}_1 \in\GG^{m\times n}}_{*, i}. 
    \\ = & 
    \set{\underbrace{\textbf{r}_2\cdot\set{\textbf{x}^\prime}_{*, 1}^T}_\text{a scalar}\|\cdots\|\underbrace{\textbf{r}_2\cdot\set{\textbf{x}^\prime}_{*, m}^T}_\text{a scalar}}\cdot \set{\textbf{r}_1}_{*,i}
    \\ = &
    \set{\textbf{r}_1}_{1,i}\cdot\textbf{r}_2\cdot\set{\textbf{x}^\prime}_{*, 1}^T+\cdots+\set{\textbf{r}_1}_{m,i}\cdot\textbf{r}_2\cdot\set{\textbf{x}^\prime}_{*, m}^T.
\end{align*}
Note that the $i$-th element is correlated with the $i$-th column of $\textbf{r}_1$ and the whole $\textbf{r}_2(\textbf{x}^{\prime})^T$.
The intermediate value $\textbf{r}_2\cdot\set{\textbf{x}^\prime}_{*, i}^T$ is reduced to be a scalar. Therefore, the correlation is element-wise doable using common multiplication. The other correlated randomness can directly be computed using ring operations. We give a toy example of computing $\textbf{c}_2$ in Figure~\ref{fig2}.
In the online phase, two parties first reveal the masked inputs to each other. As for \emph{FSS-LR-V1}, we just need to get the Taylor series approximation result in the ring field. 
We show \emph{FSS-LR-V1} in Figure~\ref{fig3}.

\begin{figure*}[!ht]
    \input{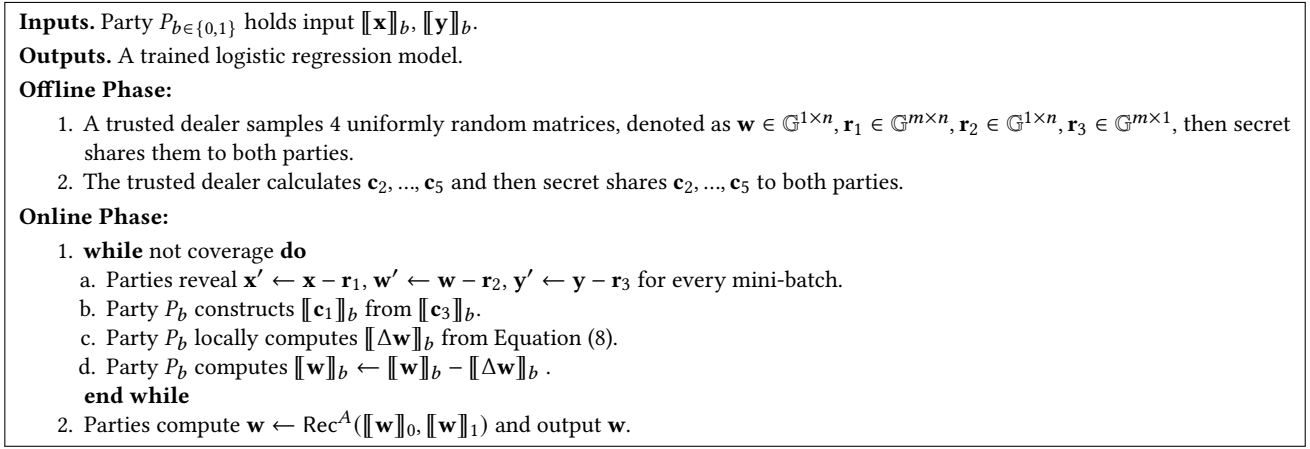}
    \caption{The logistic regression protocol (\emph{FSS-LR-V1}).}
    \label{fig3}
\end{figure*}

\begin{figure*}[!ht]
    \input{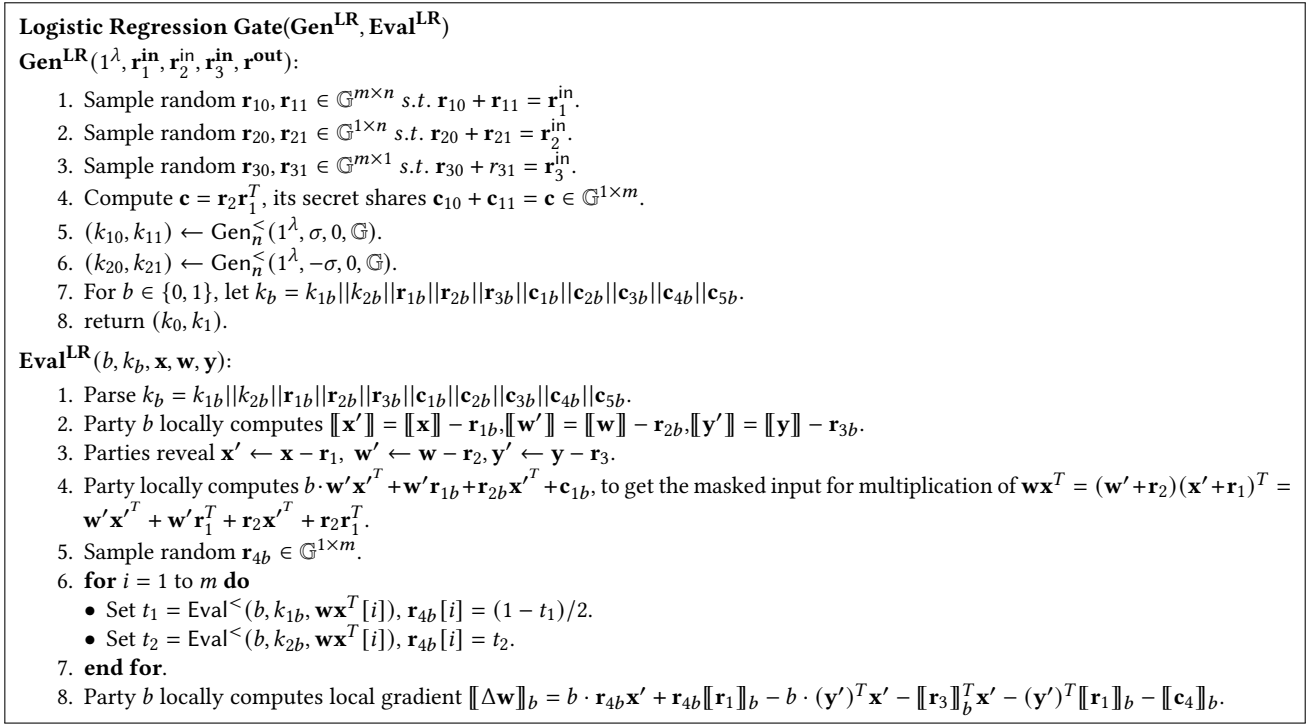}
    \caption{The logistic regression protocol (\emph{FSS-LR-V2}).}
    \label{fig4}
\end{figure*}

In \textit{FSS-LR-V2} (Figure~\ref{fig4}), the dealer will generate two DCF keys for implementing the MIC gate additionally. The weight's backpropagation continues after getting keys generated from the MIC scheme to identify intervals of the sigmoid function's different approximating modes. The model will iteratively update for some epochs. More specifically, the update weights should be in one of the below cases,
%
%

\emphsection{Case 1:} if the result of the forward propagation lies in $(-\infty,-\epsilon]$, the update weights are,
\begin{align}
    \llbracket\Delta\textbf{w}\rrbracket_b=-\frac{\alpha}{m}
     ((\textbf{y}^\prime)^T\textbf{x}^\prime
     +(\textbf{y}^\prime)^T\llbracket\textbf{r}_1\rrbracket_b
     +\llbracket\textbf{r}_3\rrbracket_b^T\textbf{x}^\prime + \llbracket\textbf{c}_4\rrbracket_b)
     \label{w_1}
\end{align}

\emphsection{Case 2:} if the result lies in $(-\epsilon,\epsilon]$, the update weights are
\begin{equation}
    \begin{aligned}
        \llbracket\Delta\textbf{w}\rrbracket_b= & \frac{\alpha}{4m}\Big(
        b\cdot2\cdot\mathbbm{1}\textbf{x}^\prime + 2\cdot\mathbbm{1}\llbracket\textbf{r}_1\rrbracket_b+ \llbracket\textbf{c}_3\rrbracket_b \\
        &+
        (
        b\cdot\textbf{w}^\prime(\textbf{x}^\prime)^T\textbf{x}^\prime 
        + \textbf{w}^\prime\llbracket\textbf{r}_1\rrbracket_b^T\textbf{x}^\prime 
        + \llbracket\textbf{r}_2\rrbracket_b(\textbf{x}^\prime)^T\textbf{x}^\prime\\
        &+ \textbf{w}^\prime(\textbf{x}^\prime)^T\llbracket\textbf{r}_1\rrbracket_b
        + \textbf{w}^\prime\llbracket\textbf{c}_5\rrbracket_b
        + \llbracket\textbf{c}_2\rrbracket_b(\textbf{x}^\prime)^T+ \llbracket\textbf{c}_1\rrbracket_b\textbf{x}^\prime)\\
        &- 4( b \cdot
        (\textbf{y}^\prime)^T\textbf{x}^\prime + (\textbf{y}^\prime)^T\llbracket\textbf{r}_1\rrbracket_b + \llbracket\textbf{r}_3\rrbracket_b^T\textbf{x}^\prime + \llbracket\textbf{c}_4\rrbracket_b
        )\Big)
    \end{aligned}
    \label{w_2}
\end{equation}

\emphsection{Case 3:} if the result lies in $(\epsilon,+\infty)$, we have
\begin{equation}
    \begin{aligned}
       \llbracket\Delta\textbf{w}\rrbracket_b= \frac{\alpha}{m}(
       &\textbf{x}^\prime-(\textbf{y}^\prime)^T\textbf{x}^\prime\\
       &+\llbracket\textbf{r}_1\rrbracket_b-(\textbf{y}^\prime)^T\llbracket\textbf{r}_1\rrbracket_b-\llbracket\textbf{r}_3\rrbracket_b^T\textbf{x}^\prime - \llbracket\textbf{c}_4\rrbracket_b)
    \end{aligned}
    \label{w_3}
\end{equation}

Now, we provide a formal security proof for our proposed protocols. Note that since we leverage standard FSS-components, the proof is intuitive and straight-forward.

\begin{lemma}
    Our protocols in Figure~\ref{fig3} and Figure~\ref{fig4} are secure instantiations (see Definition~\ref{def:simulation-security}) of LR functionality in the semi-honest treat model, given the existence of a trusted dealer who distributes correlated randomnesses.
\end{lemma}

\begin{proof}
    We prove the security of our protocol using the simulation-based technique. Since our protocol assumes the existence of a trusted dealer distributing correlated randomnesses, therefore we only need to focus on the received messages (a.k.a. the ``view'') at computing parties. Since both parties have identical processes, we only show the security for $P_{b\in\bin}$.

    In \textit{FSS-LR-V1}, the view of $P_b$ includes (1) the correlated randomness sent by the dealer: $\llbracket\textbf{w}\rrbracket_b, \llbracket\textbf{r}_1\rrbracket_b, \llbracket\textbf{r}_2\rrbracket_b, \llbracket\textbf{r}_3\rrbracket_b, \llbracket\textbf{c}_2\rrbracket_b, \cdots, \llbracket\textbf{c}_5\rrbracket_b$, (2) the revealed $\textbf{x}^\prime, \textbf{w}^\prime, \textbf{y}^\prime$. Since all secret-shared values $\llbracket\cdot\rrbracket$ are independently random, and $\textbf{x}^\prime, \textbf{w}^\prime, \textbf{y}^\prime$ are masked with uniform random $\textbf{r}_1, \textbf{r}_2, \textbf{r}_3$. Thus it's intuitive that the view is information-theoretically indistinguishable from uniform randomness.
    Similarly, in \textit{FSS-LR-V2}, parties additionally receive the DCF keys. By definition, those DCF keys are independently random (against computationally-bounded adversaries), and thus the view is indistinguishable from uniform randomness.

\end{proof}

\section{Implementations and Comparisons}
We implement our framework with C++ in $\mathbb{Z}_{2^\ell}$ and it has been observed that the modulo operations are free on CPUs~\cite{huang2022cheetah}. The ring structure is utilized to store matrices and perform fundamental operations such as addition, subtraction, multiplication, and division efficiently. 
The offline phase involves generating auxiliary random matrices to mask the original private computation data, which are created using a pseudorandom generator (PRG). For consistency, we fixed random seeds for the PRG during the experiment.

\emphsection{Experiment settings.}
The experiments were conducted on a Linux machine with an Intel(R) Xeon(R) Silver 4210R CPU @ 2.40GHz, and a cache size of 14080 KB. Two terminals were opened on a LAN network to represent the participating servers. Both servers completed the logistic regression training process with their respective private data. We implemented the logistic regression gate, which takes private data as input and outputs one round of gradients after forward and backward propagation.

\emphsection{Datasets.} 
For different experimental purposes, we sampled different datasets. We also conducted experiments on real-world datasets, such as Iris, Diabetes, and Breast Cancer.
(1) Iris dataset contains 3 classes of 50 instances each, where each class refers to a type of iris plant. We select two classes to participate in the binary logistic regression training. Four features were measured from each sample: the length and the width of the sepals and petals, in centimeters. (2) Diabetes dataset, originated from the National Institute of Diabetes and Digestive and Kidney Diseases, contains instances of 768 women from a population near Phoenix, Arizona, USA. There are 8 attributes and 1 outcome indicating the positive/negative (0/1) samples. (3) Breast Cancer dataset comes from Scikit-learn's machine learning library and contains malignant/benign (1/0) category-type data for breast cancer for 569 patients, and corresponding physiological index data for 30 dimensions. 
We distributed all datasets in an additive secret sharing way to both servers.


\emphsection{Implementation details.} 
We have developed a privacy-preserving logistic regression system using the SPU library from SecretFlow~\cite{spu}\footnote{\url{https://github.com/secretflow/spu}}, which provides a framework for provable, measurable secure computation while maintaining data privacy. To assess the accuracy and efficiency of various approximations of the sigmoid function, we conducted several experiments that included implementations of the sigmoid approximation function in plaintext, secure multi-party computation, and secure logistic regression training. We compared our FSS-LR approach with secure logistic regression implemented in MP-SPDZ~\cite{keller2020mp}\footnote{\url{https://github.com/data61/MP-SPDZ}} and SS-LR implemented in SecretFlow. Additionally, to demonstrate online efficiency, we compared the online communication time and communication overheads of SS-LR and FSS-LR on data sets of different sizes.

\subsection{Compare with SS-LR}
Using the FSS scheme, we can get the offset function family of a logistic regression function. More correlated randomness can be generated beforehand. In SS-LR, beaver triples are only used for masking the original input, whereas, in FSS-LR, correlated randomness includes all the input-independent computational terms that will take effect in secure logistic regression training.
For the case of SS-LR, the first matrix multiplication ($\textbf{w}\textbf{x}^T$) needs to open $mn+n$ elements for $\textbf{w}^\prime\in\mathbb{G}^{n\times 1}$ and $\textbf{x}^\prime\in\mathbb{G}^{m\times n}$.
Then for the second matrix multiplication, since we already have the opened value for $\textbf{x}$, we only need $m$ masked elements of the subtraction between the true and predicted labels. Therefore, the whole matrix multiplications for logistic regression need a total of $mn+n+m$ elements within $2-$round communication.
For the case of \textit{FSS-LR-V1}, we need to open $\textbf{x}^\prime\in\mathbb{G}^{m\times n}$, $\textbf{w}^\prime\in\mathbb{G}^{1\times n}$, $\textbf{y}^\prime\in\mathbb{G}^{m\times 1}$.
The whole online communication costs $mn+n+m$ elements within $1-$round communication.
Similarly, the whole online communication cost for \textit{FSS-LR-V2} is $mn+n+2m$ within $2-$round communication.

%
For the offline phase, if we allow each party to specify their own randomness shares (e.g., $\llbracket\textbf{r}\rrbracket_b$) using PRG, then the only thing the dealer needs to distribute is the beaver triples or the correlated randomness. In SS-LR, the communication overhead of the offline phase is $2mn+2n+2m$, while in \textit{FSS-LR-V1}, it is $mn+n^2+4n+3m$. The offline overhead of \textit{FSS-LR-V2} requires an additional key generation overhead of DCF, in addition to that of \textit{FSS-LR-V1}.

%
Although the costs of the offline phase in \textit{FSS-LR-V1} and \textit{FSS-LR-V2} are larger than SS-LR, the efficiency of the online phase is greatly improved. Furthermore, \textit{FSS-LR-V2} provides higher accuracy than the basic SS-LR.

 \begin{table*}[ht]
    \begin{tabular}{@{}cccccccc@{}}
    \toprule
    & \textbf{Plaintext}& \textbf{Segmented-NonLinear} & \textbf{Segmented-Taylor} & \textbf{Reciprocal} & \textbf{Square Root} & \textbf{Taylor Series} & \textbf{Sigmoid}\\
    \midrule
    Plaintext Error & - & \textbf{0.0062} & 0.0345 & 0.0577 & 0.0255 & 0.8204 & - \\
    MPC Error & - & \textbf{0.0062} & 0.0345 & 0.0577 & 0.0255 & 0.8204 & \textbf{0.0003} \\
    MPC Time(s) & - & 0.1043 & 0.0098 & \textbf{0.0092} & 0.0094 & 0.0301 & 0.2150 \\

    \midrule
   LR AUC & \textbf{0.9015} & 0.8969 & 0.8981 & \textbf{0.8996}  & 0.8922 & 0.8827 & 0.8972\\
   LR  Train Time(s) & \textbf{8.03} & 177.18 & 133.07 & 186.32 & 160.74 & \textbf{88.88} & 188.34\\
    LR Inference Time(s)&  \textbf{0.47} & 27.37 & 23.71 & \textbf{23.18} & 25.74 & 23.79 & 24.65\\
    \bottomrule
    \end{tabular}
    \vskip 5pt
    \caption{Efficiency and accuracy of different secure sigmoid approximations.}
    \label{table4}
\end{table*}

\subsection{Sigmoid Experiments}
The computation of the sigmoid function is a crucial and challenging aspect of secure logistic regression training. To enhance its effectiveness and efficiency, we conducted a comparison study of five commonly used sigmoid approximations. Our analysis included the Taylor Series expansion, two segmented functions with Taylor expansion and non-linear functions, and two complex functions comprising reciprocal and square root functions. 

\begin{enumerate}
    \item First-order Taylor expansion : $y(x)=0.5+0.25x$
    \item Segmented function with Taylor expansion: 
    \begin{align*}
        y(x)=
        \begin{cases}
        0.5 + 0.125x  & {-4 \leq x \leq -4}\\
        1 &  {x > 4}\\
        0 &  {x<-4}
        \end{cases} 
    \end{align*}

    \item Segmented function with nonlinear functions:
        \begin{align*}
            y(x)=
            \begin{cases}
            0 & {x \leq -4}\\
            1 &  {x > 4}\\
            0.5(1-0.25|x|)^2 & {-4< x \leq 0}\\
            1-0.5(1-0.25|x|)^2  & {0<x\leq 4}
            \end{cases} 
        \end{align*}

    \item Reciprocal approximation:
    $y(x)=0.5(\frac{x}{1+|x|})+0.5$

    \item Square root approximation:
    $y(x)=0.5(\frac{x}{\sqrt{1+x^2}})+0.5$
    
\end{enumerate}

\emphsection{Sigmoid accuracy.} 
 We uniformly sampled 1,000 instances between -10 and 10, and compared different approximation methods with the real sigmoid function. From Figure~\ref{fig5}, it can be seen that except for the Taylor Series approximation, the general trend of the other functions is similar to that of the real sigmoid function.
\begin{figure}[h!]
  \centering
  \includegraphics[width=0.85\linewidth]{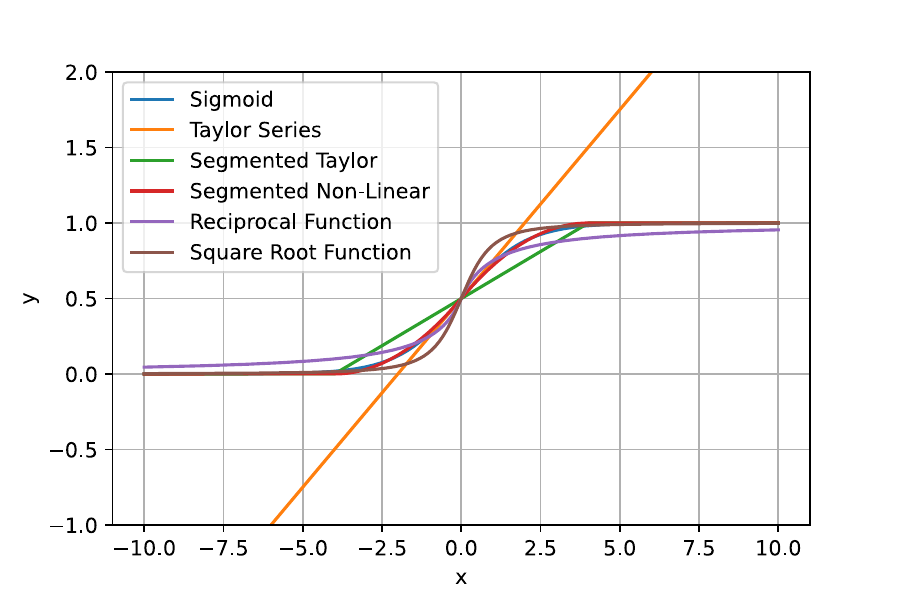}
  \caption{Sigmoid function and approximation function}
  \label{fig5}
\end{figure}

Therefore, we conducted accuracy and performance tests on various sigmoid approximations implemented in plaintext and secure multiparty computation. Specifically, we calculated their average absolute error compared to the plaintext sigmoid function and the runtime in secure multiparty computation. The upper half of Table ~\ref{table4} presents the accuracy and performance comparison between our approximate methods and the original sigmoid function. According to the results in the upper half of Table ~\ref{table4}, the MPC error is almost the same as the plaintext error in the small dataset that uses only sigmoid approximations.
This is because theoretically, the gap between MPC and plaintext mainly comes from the difference in the underlying data representation, with MPC supporting fixed-point operations and plaintext supporting floating-point operations.
The segmented function with the non-linear method exhibits the lowest error, with a value of 0.0062 for both plaintext and MPC. The reciprocal approximation function method has a slightly higher error of 0.0577 in both cases. Regarding computational time, the segmented function with the Taylor series and reciprocal approximation methods have the lowest time (9.8 and 9.2 milliseconds, respectively), followed by the square root method with 9.4 milliseconds. The segmented function with the non-linear method and Taylor Series method have higher time (104.3 and 30.1 milliseconds, respectively), and the real sigmoid method has the highest time at 215 milliseconds. Different approximation methods have their own applicable scenarios and cases. Therefore, in the next section, we will explore an appropriate method used in secure logistic regression training.

%

\emphsection{Sigmoid performance evaluation.} 
We sampled a larger dataset comprising 500,000 unique instances, each containing 100 features. To enable secure computation, we additively secret shared this dataset with two parties. Two parties then conducted logistic regression training and inference using different sigmoid approximations, through secure multi-party computation.
The lower half of Table~\ref{table4} shows that the reciprocal approximation method achieves the best approximation of the plaintext method, with a loss of 0.0018947 compared with the plaintext AUC. The Taylor series approximation method shows a marginal difference in the AUC score and inference time compared to the reciprocal approximation method. But the Taylor series approximation method significantly reduces training time, offering almost a 2x speedup compared to other approximation methods. Therefore, we will apply the Taylor series approximation in the subsequent experiments. The following \textit{FSS-LR} refers to \textit{FSS-LR-V1}.

\subsection{Logistic Regression Experiments} 
 As analyzed in Section 5.1, SS-LR requires two matrix multiplications to be completed, which involves two online interactions. In contrast, FSS-LR only requires one call to our logistic regression gate and only one round interaction to be completed throughout the gate. Here, we use real datasets mentioned at the beginning of the section to evaluate the secure logistic regression both on secret sharing with beaver triples and function secret sharing with correlated randomness. 

\begin{table}[t!]
\centering
\begin{tabular}{@{}lccc@{}}
\toprule
 & LR & SS-LR & FSS-LR   \\
 \midrule
Iris &1  & 1 & 0.9968\\
Diabetes & 0.727275 & 0.696761 & 0.676993\\
Breast Cancer &0.999943& 0.991967 & 0.983154 \\
\bottomrule
\end{tabular}
\vskip 5pt
\caption{AUC comparison of SS-LR, FSS-LR with Taylor series approximation implemented in SecretFlow and LR in plaintext.
}
\label{table5}
\end{table}

 \begin{table*}[t!]
\centering
\begin{tabular}{@{}ccc|cccc|cccc@{}}
\toprule
 & &Dataset &  &  & SS-LR &  &  &  & FSS-LR& \\
n & d & size(MB) & online(s) & total(s) & \textbf{Comm.(MB)} & rounds & online(s) & total(s) & \textbf{Comm.(MB)} & rounds \\
\midrule
 & 100 & 0.94  & 0.0365 & 0.1146 & 9.06 ($\approx 3.2\times$) & 14 & 0.0248 & 0.1697 & 2.81 & 7 \\
1000 &500  & 4.68 & 0.0893 & 0.7401 & 44.73 ($\approx 3.2\times$) & 14 & 0.0632 & 1.1808 & 13.83 & 7  \\
 & 1000 & 9.36 & 0.2295 & 1.6613 & 89.33 ($\approx 3.2\times$) & 14 & 0.1018 & 2.9279 & 27.61 & 7   \\
 \midrule
 & 100 & 9.39   & 0.4553 & 1.2044 & 1203.04 ($\approx 3.9\times$) & 156 & 0.2535 & 1.5314 & 309.27  & 78 \\
10000 & 500 & 46.83 & 1.7509 & 7.6318 & 5944.03 ($\approx 3.9\times$)  & 156 & 0.7282 & 10.7458 & 1522.19 & 78 \\
 & 1000 & 93.64 & 3.2473 & 13.9655 & 11870.26 ($\approx 3.9\times$) & 156 & 1.1400 & 28.7184 & 3038.35 & 78 \\
 \midrule
 & 100 & 93.90 & 4.3851 & 10.6551 & 121312.39 ($\approx 4.0\times$) & 1560 & 2.5071 & 14.0974 & 30653.12 & 780 \\
100000 & 500  & 468.34 & 17.6945 & 56.3174 & 599398.41 ($\approx 4.0\times$) & 1560 & 7.3247 & 92.4672 & 150870.59 & 780 \\
 & 1000 & 936.37 & 32.7224 & 112.5950 & 1197005.94 ($\approx 4.0\times$) & 1560 & 11.5450 & 255.4590 & 301142.42 & 780 \\
\bottomrule
\end{tabular}
\vskip 5pt
\caption{Performance of the online phase. batch size = 128, epoch = 1, learning rate = 0.5. 
}
\label{table6}
\end{table*}

\emphsection{Accuracy evaluation.} The converged AUCs of logistic regression training in plaintext, secret sharing, and function secret sharing are shown in Table~\ref{table5}. The AUC of our scheme is very close to that of doing logistic regression in plaintext and secret sharing. However, there is still a gap in our approach compared with the plaintext LR and SS-LR. The AUC loss of our method is mainly because we designed the logistic regression gate with multiple truncations inside, each of which introduces some small, unavoidable errors. 

\begin{table}[!h]
\centering
\begin{tabular}{@{}lccc@{}}
\toprule
 & MP-SPDZ & SS-LR & FSS-LR   \\
 \midrule
Online time(s)& 11.2792 & 0.0305 & 0.0289\\
Total time(s) & 12.2984& 0.0703 & 0.0520 \\
Total rounds & 16608 & 2 & 1\\
Online Communication(MB) & 92.7595 & 11.6 & 2.6 \\
\bottomrule
\end{tabular}
\vskip 5pt
\caption{Online effectiveness and efficiency comparison of MP-SPDZ, SS-LR, and FSS-LR in the breast cancer dataset. }
\label{table7}
\end{table}

\emphsection{Performance evaluation.} Our experiments evaluate the online phase performance of FSS-LR compared with the baseline experiments SS-LR and secure logistic regression training using SGD implemented in MP-SPDZ. 
The online phase consists of data-dependent communication between two parties. In our experiments, we fix the training parameters and count the online communication time, rounds, overheads, and total running time. 
In Table ~\ref{table7}, we compare the online performance of MP-SPDZ, SS-LR, and FSS-LR on the same dataset in 1 epoch. 
Based on Table ~\ref{table7}, we can observe that the FSS-LR method demonstrates higher efficiency and better performance in terms of online time, total time, total rounds, and online communication compared to the MP-SPDZ and SS-LR methods. Specifically, in terms of online time and total time, the FSS-LR method is respectively 390.7 and 26.9 times faster than the MP-SPDZ method, and 1.05 and 1.4 times faster than the SS-LR method. The FSS-LR method requires only one round of communication, whereas the MP-SPDZ method requires 16,608 rounds and the SS-LR method requires 2 rounds. Regarding online communication, the FSS-LR method requires only 2.6 MB, which corresponds to 2.8\% of the MP-SPDZ method and 22.4\% of the SS-LR method. Therefore, it can be concluded that the FSS-LR method outperforms the MP-SPDZ and SS-LR methods in terms of online performance and communication efficiency. 

To further demonstrate the scalability and the acceleration effect of our proposed method on the online phase, we randomly generated simulated datasets of various magnitudes to train the logistic regression model using SS-LR and FSS-LR. In this experiment, 3 different sizes of instances ($n = 1000, 10000, 100000$) are set and 3 dimensions ($d = 100, 500, 1000$) are designed for each size of instances.
Results are shown in Table~\ref{table6}.
 When the sampling size is relatively small, such as $1000\times100$, the difference between FSS-LR and SS-LR in terms of online phase time and the total communication volume is not very large. 
 However, as the dataset size gets larger, we can clearly find that the communication volume and communication time of SS-LR is almost three times that of FSS-LR. 
 \begin{figure}[!h]
\centering
\includegraphics[width=1\columnwidth]{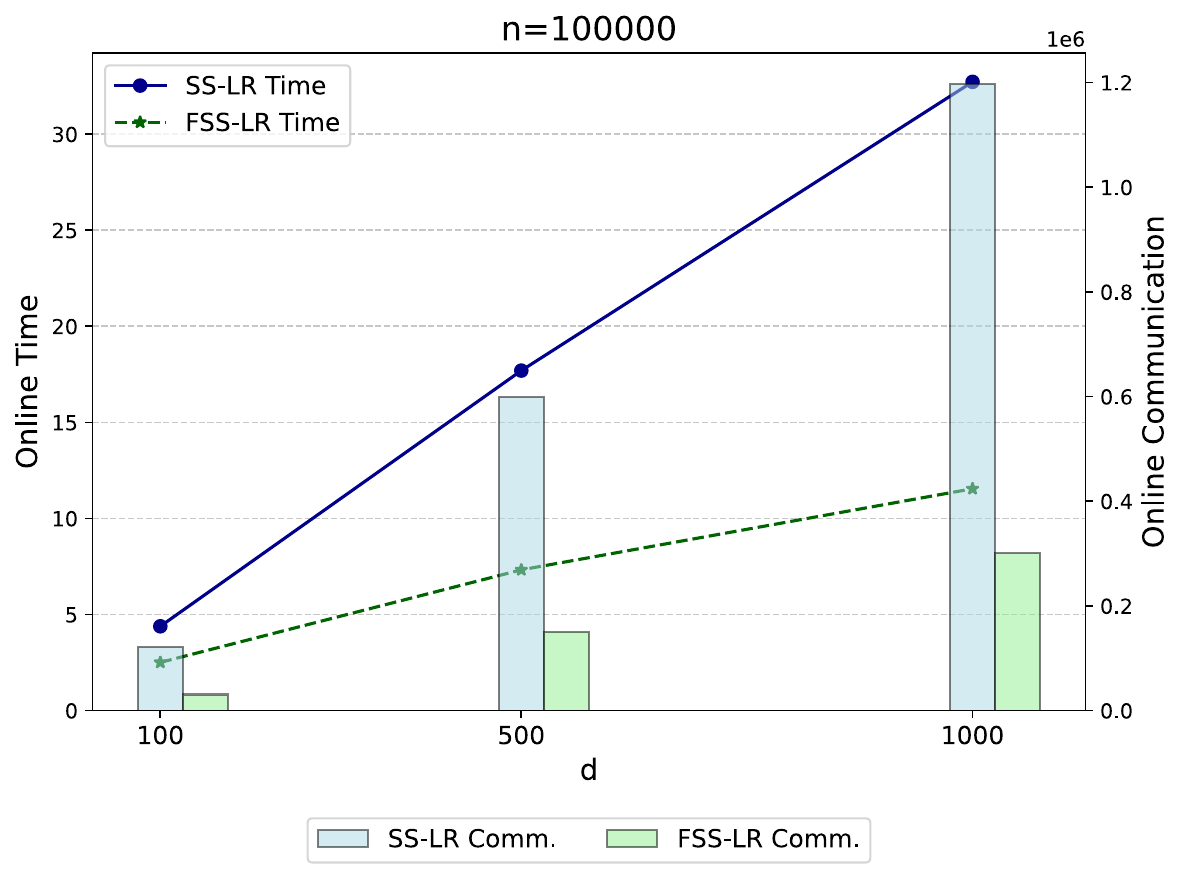} 
\caption{The online time and communication volume for different sizes of datasets.}
\vspace{-0.8cm} 
\vskip 5pt
\label{fig6}
\end{figure}

We set the number of samples $n=100000$, and the feature dimension $d=100, 500, 1000$ and show the performance of SS-LR and FSS-LR more visually in Figure ~\ref{fig6}. The online time of SS-LR is almost three times that of FSS-LR, and the communication load is about 3 to 4 times. From the above experiments, we show that our method trains a logistic regression model with efficient function secret sharing, and performs optimal online communication.

\section{Conclusion}

In this work, we leverage function secret sharing techniques to accelerate the online communication phase of the secure logistic regression training. 
We design MPC-friendly approximation schemes for the activation function and implement a complete online optimal logistic regression gate using function secret sharing techniques, which accepts masked input and outputs masked results, enabling two parties to decrypt the final resulting model at the end of training.
Our scheme significantly improves the online communication efficiency compared to the original SS-LR.
However, there is still room for improvement in terms of overall overhead and model accuracy, especially in large datasets and complex models.

\section{Acknowledgement}
This work was supported by the National Natural Science Foundation of China under grant number 62202170 and the Ant Group.


\pagebreak
\bibliographystyle{ACM-Reference-Format}
\bibliography{ref}


\begin{thebibliography}{48}


\ifx \showCODEN    \undefined \def \showCODEN     #1{\unskip}     \fi
\ifx \showDOI      \undefined \def \showDOI       #1{#1}\fi
\ifx \showISBNx    \undefined \def \showISBNx     #1{\unskip}     \fi
\ifx \showISBNxiii \undefined \def \showISBNxiii  #1{\unskip}     \fi
\ifx \showISSN     \undefined \def \showISSN      #1{\unskip}     \fi
\ifx \showLCCN     \undefined \def \showLCCN      #1{\unskip}     \fi
\ifx \shownote     \undefined \def \shownote      #1{#1}          \fi
\ifx \showarticletitle \undefined \def \showarticletitle #1{#1}   \fi
\ifx \showURL      \undefined \def \showURL       {\relax}        \fi
\providecommand\bibfield[2]{#2}
\providecommand\bibinfo[2]{#2}
\providecommand\natexlab[1]{#1}
\providecommand\showeprint[2][]{arXiv:#2}

\bibitem[Acar et~al\mbox{.}(2018)]%
        {acar2018survey}
\bibfield{author}{\bibinfo{person}{Abbas Acar}, \bibinfo{person}{Hidayet Aksu},
  \bibinfo{person}{A~Selcuk Uluagac}, {and} \bibinfo{person}{Mauro Conti}.}
  \bibinfo{year}{2018}\natexlab{}.
\newblock \showarticletitle{A survey on homomorphic encryption schemes: Theory
  and implementation}.
\newblock \bibinfo{journal}{\emph{ACM Computing Surveys (Csur)}}
  \bibinfo{volume}{51}, \bibinfo{number}{4} (\bibinfo{year}{2018}),
  \bibinfo{pages}{1--35}.
\newblock


\bibitem[Agrawal et~al\mbox{.}(2019)]%
        {agrawal2019quotient}
\bibfield{author}{\bibinfo{person}{Nitin Agrawal}, \bibinfo{person}{Ali~Shahin
  Shamsabadi}, \bibinfo{person}{Matt~J. Kusner}, {and}
  \bibinfo{person}{Adri{\`{a}} Gasc{\'{o}}n}.} \bibinfo{year}{2019}\natexlab{}.
\newblock \showarticletitle{{QUOTIENT:} Two-Party Secure Neural Network
  Training and Prediction}. In \bibinfo{booktitle}{\emph{Proceedings of the
  2019 {ACM} {SIGSAC} Conference on Computer and Communications Security, {CCS}
  2019, London, UK, November 11-15, 2019}}. \bibinfo{publisher}{{ACM}},
  \bibinfo{pages}{1231--1247}.
\newblock
\urldef\tempurl%
\url{https://doi.org/10.1145/3319535.3339819}
\showDOI{\tempurl}


\bibitem[Aono et~al\mbox{.}(2016)]%
        {aono2016scalable}
\bibfield{author}{\bibinfo{person}{Yoshinori Aono}, \bibinfo{person}{Takuya
  Hayashi}, \bibinfo{person}{Le Trieu~Phong}, {and} \bibinfo{person}{Lihua
  Wang}.} \bibinfo{year}{2016}\natexlab{}.
\newblock \showarticletitle{Scalable and secure logistic regression via
  homomorphic encryption}. In \bibinfo{booktitle}{\emph{Proceedings of the
  Sixth ACM Conference on Data and Application Security and Privacy}}.
  \bibinfo{pages}{142--144}.
\newblock


\bibitem[Baum et~al\mbox{.}(2016)]%
        {SMPC}
\bibfield{author}{\bibinfo{person}{Carsten Baum}, \bibinfo{person}{Ivan
  Damg{\aa}rd}, \bibinfo{person}{Tomas Toft}, {and}
  \bibinfo{person}{Rasmus~Winther Zakarias}.} \bibinfo{year}{2016}\natexlab{}.
\newblock \showarticletitle{Better Preprocessing for Secure Multiparty
  Computation}. In \bibinfo{booktitle}{\emph{International Conference on
  Applied Cryptography and Network Security}}.
\newblock


\bibitem[Beaver(1991)]%
        {Beaver91a}
\bibfield{author}{\bibinfo{person}{Donald Beaver}.}
  \bibinfo{year}{1991}\natexlab{}.
\newblock \showarticletitle{Efficient Multiparty Protocols Using Circuit
  Randomization}. In \bibinfo{booktitle}{\emph{Annual International Cryptology
  Conference}}.
\newblock


\bibitem[Boyle et~al\mbox{.}(2021)]%
        {boyle2021function}
\bibfield{author}{\bibinfo{person}{Elette Boyle}, \bibinfo{person}{Nishanth
  Chandran}, \bibinfo{person}{Niv Gilboa}, \bibinfo{person}{Divya Gupta},
  \bibinfo{person}{Yuval Ishai}, \bibinfo{person}{Nishant Kumar}, {and}
  \bibinfo{person}{Mayank Rathee}.} \bibinfo{year}{2021}\natexlab{}.
\newblock \showarticletitle{Function secret sharing for mixed-mode and
  fixed-point secure computation}. In \bibinfo{booktitle}{\emph{Annual
  International Conference on the Theory and Applications of Cryptographic
  Techniques}}. Springer, \bibinfo{pages}{871--900}.
\newblock


\bibitem[Boyle et~al\mbox{.}(2015)]%
        {boyle2015function}
\bibfield{author}{\bibinfo{person}{Elette Boyle}, \bibinfo{person}{Niv Gilboa},
  {and} \bibinfo{person}{Yuval Ishai}.} \bibinfo{year}{2015}\natexlab{}.
\newblock \showarticletitle{Function secret sharing}. In
  \bibinfo{booktitle}{\emph{Annual international conference on the theory and
  applications of cryptographic techniques}}. Springer,
  \bibinfo{pages}{337--367}.
\newblock


\bibitem[Boyle et~al\mbox{.}(2016)]%
        {boyle2016function}
\bibfield{author}{\bibinfo{person}{Elette Boyle}, \bibinfo{person}{Niv Gilboa},
  {and} \bibinfo{person}{Yuval Ishai}.} \bibinfo{year}{2016}\natexlab{}.
\newblock \showarticletitle{Function secret sharing: Improvements and
  extensions}. In \bibinfo{booktitle}{\emph{Proceedings of the 2016 ACM SIGSAC
  Conference on Computer and Communications Security}}.
  \bibinfo{pages}{1292--1303}.
\newblock


\bibitem[Boyle et~al\mbox{.}(2019)]%
        {DBLP:conf/tcc/BoyleGI19}
\bibfield{author}{\bibinfo{person}{Elette Boyle}, \bibinfo{person}{Niv Gilboa},
  {and} \bibinfo{person}{Yuval Ishai}.} \bibinfo{year}{2019}\natexlab{}.
\newblock \showarticletitle{Secure computation with preprocessing via function
  secret sharing}. In \bibinfo{booktitle}{\emph{Theory of Cryptography}}.
  Springer.
\newblock


\bibitem[Byali et~al\mbox{.}(2019)]%
        {byali2019flash}
\bibfield{author}{\bibinfo{person}{Megha Byali}, \bibinfo{person}{Harsh
  Chaudhari}, \bibinfo{person}{Arpita Patra}, {and} \bibinfo{person}{Ajith
  Suresh}.} \bibinfo{year}{2019}\natexlab{}.
\newblock \showarticletitle{FLASH: fast and robust framework for
  privacy-preserving machine learning}.
\newblock \bibinfo{journal}{\emph{Cryptology ePrint Archive}}
  (\bibinfo{year}{2019}).
\newblock


\bibitem[Chaudhari et~al\mbox{.}(2019)]%
        {chaudhari2019astra}
\bibfield{author}{\bibinfo{person}{Harsh Chaudhari}, \bibinfo{person}{Ashish
  Choudhury}, \bibinfo{person}{Arpita Patra}, {and} \bibinfo{person}{Ajith
  Suresh}.} \bibinfo{year}{2019}\natexlab{}.
\newblock \showarticletitle{ASTRA: High Throughput 3PC over Rings with
  Application to Secure Prediction}.
\newblock \bibinfo{journal}{\emph{Proceedings of the 2019 ACM SIGSAC Conference
  on Cloud Computing Security Workshop}} (\bibinfo{year}{2019}).
\newblock


\bibitem[Cheng et~al\mbox{.}(2016)]%
        {cheng2016wide}
\bibfield{author}{\bibinfo{person}{Heng-Tze Cheng}, \bibinfo{person}{Levent
  Koc}, \bibinfo{person}{Jeremiah Harmsen}, \bibinfo{person}{Tal Shaked},
  \bibinfo{person}{Tushar Chandra}, \bibinfo{person}{Hrishi Aradhye},
  \bibinfo{person}{Glen Anderson}, \bibinfo{person}{Greg Corrado},
  \bibinfo{person}{Wei Chai}, \bibinfo{person}{Mustafa Ispir}, {et~al\mbox{.}}}
  \bibinfo{year}{2016}\natexlab{}.
\newblock \showarticletitle{Wide \& deep learning for recommender systems}. In
  \bibinfo{booktitle}{\emph{Proceedings of the 1st workshop on deep learning
  for recommender systems}}. \bibinfo{pages}{7--10}.
\newblock


\bibitem[Dalskov et~al\mbox{.}(2021)]%
        {dalskov2021fantastic}
\bibfield{author}{\bibinfo{person}{Anders~PK Dalskov}, \bibinfo{person}{Daniel
  Escudero}, {and} \bibinfo{person}{Marcel Keller}.}
  \bibinfo{year}{2021}\natexlab{}.
\newblock \showarticletitle{Fantastic Four: Honest-Majority Four-Party Secure
  Computation With Malicious Security.}. In \bibinfo{booktitle}{\emph{USENIX
  Security Symposium}}. \bibinfo{pages}{2183--2200}.
\newblock


\bibitem[Demmler et~al\mbox{.}(2015)]%
        {demmler2015aby}
\bibfield{author}{\bibinfo{person}{Daniel Demmler}, \bibinfo{person}{Thomas
  Schneider}, {and} \bibinfo{person}{Michael Zohner}.}
  \bibinfo{year}{2015}\natexlab{}.
\newblock \showarticletitle{ABY-A framework for efficient mixed-protocol secure
  two-party computation.}. In \bibinfo{booktitle}{\emph{NDSS}}.
\newblock


\bibitem[Dwork(2008)]%
        {dwork2008differential}
\bibfield{author}{\bibinfo{person}{Cynthia Dwork}.}
  \bibinfo{year}{2008}\natexlab{}.
\newblock \showarticletitle{Differential privacy: A survey of results}. In
  \bibinfo{booktitle}{\emph{International conference on theory and applications
  of models of computation}}. Springer, \bibinfo{pages}{1--19}.
\newblock


\bibitem[Even et~al\mbox{.}(2005)]%
        {even2005parametric}
\bibfield{author}{\bibinfo{person}{Guy Even}, \bibinfo{person}{Peter-M Seidel},
  {and} \bibinfo{person}{Warren~E Ferguson}.} \bibinfo{year}{2005}\natexlab{}.
\newblock \showarticletitle{A parametric error analysis of Goldschmidt's
  division algorithm}.
\newblock \bibinfo{journal}{\emph{J. Comput. System Sci.}}
  \bibinfo{volume}{70}, \bibinfo{number}{1} (\bibinfo{year}{2005}),
  \bibinfo{pages}{118--139}.
\newblock


\bibitem[Fienberg et~al\mbox{.}(2006)]%
        {DBLP:conf/psd/FienbergFSW06}
\bibfield{author}{\bibinfo{person}{Stephen~E. Fienberg},
  \bibinfo{person}{William~J. Fulp}, \bibinfo{person}{Aleksandra~B. Slavkovic},
  {and} \bibinfo{person}{Tracey~A. Wrobel}.} \bibinfo{year}{2006}\natexlab{}.
\newblock \showarticletitle{"Secure" Log-Linear and Logistic Regression
  Analysis of Distributed Databases}. In \bibinfo{booktitle}{\emph{Privacy in
  Statistical Databases}}.
\newblock


\bibitem[Goldreich(1998)]%
        {goldreich1998secure}
\bibfield{author}{\bibinfo{person}{Oded Goldreich}.}
  \bibinfo{year}{1998}\natexlab{}.
\newblock \showarticletitle{Secure multi-party computation}.
\newblock \bibinfo{journal}{\emph{Manuscript. Preliminary version}}
  \bibinfo{volume}{78} (\bibinfo{year}{1998}), \bibinfo{pages}{110}.
\newblock


\bibitem[Huang et~al\mbox{.}(2022)]%
        {huang2022cheetah}
\bibfield{author}{\bibinfo{person}{Zhicong Huang}, \bibinfo{person}{Wen jie
  Lu}, \bibinfo{person}{Cheng Hong}, {and} \bibinfo{person}{Jiansheng Ding}.}
  \bibinfo{year}{2022}\natexlab{}.
\newblock \showarticletitle{Cheetah: Lean and Fast Secure {Two-Party} Deep
  Neural Network Inference}. In \bibinfo{booktitle}{\emph{31st USENIX Security
  Symposium (USENIX Security 22)}}. \bibinfo{publisher}{USENIX Association},
  \bibinfo{address}{Boston, MA}, \bibinfo{pages}{809--826}.
\newblock
\showISBNx{978-1-939133-31-1}


\bibitem[Juvekar et~al\mbox{.}(2018)]%
        {gazelle}
\bibfield{author}{\bibinfo{person}{Chiraag Juvekar}, \bibinfo{person}{Vinod
  Vaikuntanathan}, {and} \bibinfo{person}{Anantha Chandrakasan}.}
  \bibinfo{year}{2018}\natexlab{}.
\newblock \showarticletitle{{GAZELLE}: A Low Latency Framework for Secure
  Neural Network Inference}. In \bibinfo{booktitle}{\emph{27th USENIX Security
  Symposium (USENIX Security 18)}}. \bibinfo{publisher}{USENIX Association},
  \bibinfo{address}{Baltimore, MD}, \bibinfo{pages}{1651--1669}.
\newblock
\showISBNx{978-1-939133-04-5}


\bibitem[Keller(2020)]%
        {keller2020mp}
\bibfield{author}{\bibinfo{person}{Marcel Keller}.}
  \bibinfo{year}{2020}\natexlab{}.
\newblock \showarticletitle{MP-SPDZ: A versatile framework for multi-party
  computation}. In \bibinfo{booktitle}{\emph{Proceedings of the 2020 ACM SIGSAC
  conference on computer and communications security}}.
  \bibinfo{pages}{1575--1590}.
\newblock


\bibitem[Kim and Gu(2006)]%
        {kim2006logistic}
\bibfield{author}{\bibinfo{person}{Hyunjoon Kim} {and} \bibinfo{person}{Zheng
  Gu}.} \bibinfo{year}{2006}\natexlab{}.
\newblock \showarticletitle{A logistic regression analysis for predicting
  bankruptcy in the hospitality industry}.
\newblock \bibinfo{journal}{\emph{The Journal of Hospitality Financial
  Management}} \bibinfo{volume}{14}, \bibinfo{number}{1}
  (\bibinfo{year}{2006}), \bibinfo{pages}{17--34}.
\newblock


\bibitem[Koti et~al\mbox{.}(2021)]%
        {koti2021swift}
\bibfield{author}{\bibinfo{person}{Nishat Koti}, \bibinfo{person}{Mahak
  Pancholi}, \bibinfo{person}{Arpita Patra}, {and} \bibinfo{person}{Ajith
  Suresh}.} \bibinfo{year}{2021}\natexlab{}.
\newblock \showarticletitle{SWIFT: Super-fast and Robust Privacy-Preserving
  Machine Learning.}. In \bibinfo{booktitle}{\emph{USENIX Security Symposium}}.
  \bibinfo{pages}{2651--2668}.
\newblock


\bibitem[Lindell(2017)]%
        {DBLP:journals/eccc/Lindell17}
\bibfield{author}{\bibinfo{person}{Yehuda Lindell}.}
  \bibinfo{year}{2017}\natexlab{}.
\newblock \showarticletitle{How To Simulate It - {A} Tutorial on the Simulation
  Proof Technique}.
\newblock \bibinfo{journal}{\emph{Electron. Colloquium Comput. Complex.}}
  \bibinfo{volume}{{TR17-112}} (\bibinfo{year}{2017}).
\newblock
\showeprint[ECCC]{TR17-112}


\bibitem[Liu et~al\mbox{.}(2017)]%
        {minionn}
\bibfield{author}{\bibinfo{person}{Jian Liu}, \bibinfo{person}{Mika Juuti},
  \bibinfo{person}{Yao Lu}, {and} \bibinfo{person}{Nadarajah Asokan}.}
  \bibinfo{year}{2017}\natexlab{}.
\newblock \showarticletitle{Oblivious neural network predictions via minionn
  transformations}. In \bibinfo{booktitle}{\emph{Proceedings of the 2017 ACM
  SIGSAC conference on computer and communications security}}.
  \bibinfo{pages}{619--631}.
\newblock


\bibitem[Lu et~al\mbox{.}(2023)]%
        {lu2023squirrel}
\bibfield{author}{\bibinfo{person}{Wen-jie Lu}, \bibinfo{person}{Zhicong
  Huang}, \bibinfo{person}{Qizhi Zhang}, \bibinfo{person}{Yuchen Wang}, {and}
  \bibinfo{person}{Cheng Hong}.} \bibinfo{year}{2023}\natexlab{}.
\newblock \showarticletitle{Squirrel: A Scalable Secure Two-Party Computation
  Framework for Training Gradient Boosting Decision Tree}.
\newblock \bibinfo{journal}{\emph{Cryptology ePrint Archive}}
  (\bibinfo{year}{2023}).
\newblock


\bibitem[Ma et~al\mbox{.}(2023)]%
        {spu}
\bibfield{author}{\bibinfo{person}{Junming Ma}, \bibinfo{person}{Yancheng
  Zheng}, \bibinfo{person}{Jun Feng}, \bibinfo{person}{Derun Zhao},
  \bibinfo{person}{Haoqi Wu}, \bibinfo{person}{Wenjing Fang},
  \bibinfo{person}{Jin Tan}, \bibinfo{person}{Chaofan Yu},
  \bibinfo{person}{Benyu Zhang}, {and} \bibinfo{person}{Lei Wang}.}
  \bibinfo{year}{2023}\natexlab{}.
\newblock \showarticletitle{SecretFlow-SPU: A Performant and User-Friendly
  Framework for Privacy-Preserving Machine Learning}. In
  \bibinfo{booktitle}{\emph{USENIX Annual Technical Conference}}.
\newblock


\bibitem[Markstein(2004)]%
        {markstein2004software}
\bibfield{author}{\bibinfo{person}{Peter Markstein}.}
  \bibinfo{year}{2004}\natexlab{}.
\newblock \showarticletitle{Software division and square root using
  Goldschmidt’s algorithms}. In \bibinfo{booktitle}{\emph{Proceedings of the
  6th Conference on Real Numbers and Computers (RNC’6)}},
  Vol.~\bibinfo{volume}{123}. \bibinfo{pages}{146--157}.
\newblock


\bibitem[Mishra et~al\mbox{.}(2020)]%
        {delphi}
\bibfield{author}{\bibinfo{person}{Pratyush Mishra}, \bibinfo{person}{Ryan
  Lehmkuhl}, \bibinfo{person}{Akshayaram Srinivasan}, \bibinfo{person}{Wenting
  Zheng}, {and} \bibinfo{person}{Raluca~Ada Popa}.}
  \bibinfo{year}{2020}\natexlab{}.
\newblock \showarticletitle{Delphi: A Cryptographic Inference Service for
  Neural Networks}. In \bibinfo{booktitle}{\emph{29th USENIX Security Symposium
  (USENIX Security 20)}}. \bibinfo{publisher}{USENIX Association},
  \bibinfo{pages}{2505--2522}.
\newblock
\showISBNx{978-1-939133-17-5}


\bibitem[Mohassel and Rindal(2018)]%
        {mohassel2018aby3}
\bibfield{author}{\bibinfo{person}{Payman Mohassel} {and}
  \bibinfo{person}{Peter Rindal}.} \bibinfo{year}{2018}\natexlab{}.
\newblock \showarticletitle{ABY3: A mixed protocol framework for machine
  learning}. In \bibinfo{booktitle}{\emph{Proceedings of the 2018 ACM SIGSAC
  conference on computer and communications security}}.
  \bibinfo{pages}{35--52}.
\newblock


\bibitem[Mohassel and Zhang(2017)]%
        {secureml}
\bibfield{author}{\bibinfo{person}{Payman Mohassel} {and}
  \bibinfo{person}{Yupeng Zhang}.} \bibinfo{year}{2017}\natexlab{}.
\newblock \showarticletitle{SecureML: A System for Scalable Privacy-Preserving
  Machine Learning}.
\newblock \bibinfo{journal}{\emph{2017 IEEE Symposium on Security and Privacy
  (SP)}} (\bibinfo{year}{2017}), \bibinfo{pages}{19--38}.
\newblock


\bibitem[Noubir and Sanatinia(2016)]%
        {noubir2016trusted}
\bibfield{author}{\bibinfo{person}{Guevara Noubir} {and}
  \bibinfo{person}{Amirali Sanatinia}.} \bibinfo{year}{2016}\natexlab{}.
\newblock \showarticletitle{Trusted code execution on untrusted platforms using
  Intel SGX}.
\newblock \bibinfo{journal}{\emph{Virus bulletin}} (\bibinfo{year}{2016}).
\newblock


\bibitem[Nusinovici et~al\mbox{.}(2020)]%
        {nusinovici2020logistic}
\bibfield{author}{\bibinfo{person}{Simon Nusinovici},
  \bibinfo{person}{Yih~Chung Tham}, \bibinfo{person}{Marco Yu~Chak Yan},
  \bibinfo{person}{Daniel Shu~Wei Ting}, \bibinfo{person}{Jialiang Li},
  \bibinfo{person}{Charumathi Sabanayagam}, \bibinfo{person}{Tien~Yin Wong},
  {and} \bibinfo{person}{Ching-Yu Cheng}.} \bibinfo{year}{2020}\natexlab{}.
\newblock \showarticletitle{Logistic regression was as good as machine learning
  for predicting major chronic diseases}.
\newblock \bibinfo{journal}{\emph{Journal of clinical epidemiology}}
  \bibinfo{volume}{122} (\bibinfo{year}{2020}), \bibinfo{pages}{56--69}.
\newblock


\bibitem[Patra et~al\mbox{.}(2021)]%
        {patra2021aby2}
\bibfield{author}{\bibinfo{person}{Arpita Patra}, \bibinfo{person}{Thomas
  Schneider}, \bibinfo{person}{Ajith Suresh}, {and} \bibinfo{person}{Hossein
  Yalame}.} \bibinfo{year}{2021}\natexlab{}.
\newblock \showarticletitle{ABY2. 0: Improved Mixed-Protocol Secure Two-Party
  Computation.}. In \bibinfo{booktitle}{\emph{USENIX Security Symposium}}.
\newblock


\bibitem[Patra and Suresh(2020)]%
        {patra2020blaze}
\bibfield{author}{\bibinfo{person}{Arpita Patra} {and} \bibinfo{person}{Ajith
  Suresh}.} \bibinfo{year}{2020}\natexlab{}.
\newblock \showarticletitle{BLAZE: blazing fast privacy-preserving machine
  learning}.
\newblock \bibinfo{journal}{\emph{arXiv preprint arXiv:2005.09042}}
  (\bibinfo{year}{2020}).
\newblock


\bibitem[Rathee et~al\mbox{.}(2022)]%
        {rathee2022secfloat}
\bibfield{author}{\bibinfo{person}{Deevashwer Rathee}, \bibinfo{person}{Anwesh
  Bhattacharya}, \bibinfo{person}{Rahul Sharma}, \bibinfo{person}{Divya Gupta},
  \bibinfo{person}{Nishanth Chandran}, {and} \bibinfo{person}{Aseem Rastogi}.}
  \bibinfo{year}{2022}\natexlab{}.
\newblock \showarticletitle{SecFloat: Accurate Floating-Point meets Secure
  2-Party Computation}. In \bibinfo{booktitle}{\emph{2022 IEEE Symposium on
  Security and Privacy (SP)}}. IEEE, \bibinfo{pages}{576--595}.
\newblock


\bibitem[Rathee et~al\mbox{.}(2020)]%
        {cryptflow2}
\bibfield{author}{\bibinfo{person}{Deevashwer Rathee}, \bibinfo{person}{Mayank
  Rathee}, \bibinfo{person}{Nishant Kumar}, \bibinfo{person}{Nishanth
  Chandran}, \bibinfo{person}{Divya Gupta}, \bibinfo{person}{Aseem Rastogi},
  {and} \bibinfo{person}{Rahul Sharma}.} \bibinfo{year}{2020}\natexlab{}.
\newblock \showarticletitle{CrypTFlow2: Practical 2-party secure inference}. In
  \bibinfo{booktitle}{\emph{Proceedings of the 2020 ACM SIGSAC Conference on
  Computer and Communications Security}}. \bibinfo{pages}{325--342}.
\newblock


\bibitem[Riazi et~al\mbox{.}(2018)]%
        {Chameleon}
\bibfield{author}{\bibinfo{person}{M.~Sadegh Riazi}, \bibinfo{person}{Christian
  Weinert}, \bibinfo{person}{Oleksandr Tkachenko}, \bibinfo{person}{Ebrahim~M.
  Songhori}, \bibinfo{person}{T. Schneider}, {and} \bibinfo{person}{Farinaz
  Koushanfar}.} \bibinfo{year}{2018}\natexlab{}.
\newblock \showarticletitle{Chameleon: A Hybrid Secure Computation Framework
  for Machine Learning Applications}.
\newblock \bibinfo{journal}{\emph{Proceedings of the 2018 on Asia Conference on
  Computer and Communications Security}} (\bibinfo{year}{2018}).
\newblock


\bibitem[Ryffel et~al\mbox{.}(2020)]%
        {ariann}
\bibfield{author}{\bibinfo{person}{Th{\'e}o Ryffel}, \bibinfo{person}{Pierre
  Tholoniat}, \bibinfo{person}{David Pointcheval}, {and}
  \bibinfo{person}{Francis~R. Bach}.} \bibinfo{year}{2020}\natexlab{}.
\newblock \showarticletitle{AriaNN: Low-Interaction Privacy-Preserving Deep
  Learning via Function Secret Sharing}.
\newblock \bibinfo{journal}{\emph{Proceedings on Privacy Enhancing
  Technologies}}  \bibinfo{volume}{2022} (\bibinfo{year}{2020}),
  \bibinfo{pages}{291 -- 316}.
\newblock


\bibitem[Sahin and Duman(2011)]%
        {sahin2011detecting}
\bibfield{author}{\bibinfo{person}{Yusuf Sahin} {and} \bibinfo{person}{Ekrem
  Duman}.} \bibinfo{year}{2011}\natexlab{}.
\newblock \showarticletitle{Detecting credit card fraud by ANN and logistic
  regression}. In \bibinfo{booktitle}{\emph{2011 international symposium on
  innovations in intelligent systems and applications}}. IEEE,
  \bibinfo{pages}{315--319}.
\newblock


\bibitem[Slavkovic et~al\mbox{.}(2007)]%
        {DBLP:conf/icdm/SlavkovicNT07}
\bibfield{author}{\bibinfo{person}{Aleksandra~B. Slavkovic},
  \bibinfo{person}{Yuval Nardi}, {and} \bibinfo{person}{Matthew~M. Tibbits}.}
  \bibinfo{year}{2007}\natexlab{}.
\newblock \showarticletitle{"Secure" Logistic Regression of Horizontally and
  Vertically Partitioned Distributed Databases}.
\newblock \bibinfo{journal}{\emph{Seventh IEEE International Conference on Data
  Mining Workshops (ICDMW 2007)}} (\bibinfo{year}{2007}),
  \bibinfo{pages}{723--728}.
\newblock


\bibitem[Song et~al\mbox{.}(2022)]%
        {song2022pmpl}
\bibfield{author}{\bibinfo{person}{Lushan Song}, \bibinfo{person}{Jiaxuan
  Wang}, \bibinfo{person}{Zhexuan Wang}, \bibinfo{person}{Xinyu Tu},
  \bibinfo{person}{Guopeng Lin}, \bibinfo{person}{Wenqiang Ruan},
  \bibinfo{person}{Haoqi Wu}, {and} \bibinfo{person}{Weili Han}.}
  \bibinfo{year}{2022}\natexlab{}.
\newblock \showarticletitle{pMPL: A Robust Multi-Party Learning Framework with
  a Privileged Party}. In \bibinfo{booktitle}{\emph{Proceedings of the 2022 ACM
  SIGSAC Conference on eomputer and Communications Security}}.
  \bibinfo{pages}{2689--2703}.
\newblock


\bibitem[Tan et~al\mbox{.}(2021)]%
        {tan2021cryptgpu}
\bibfield{author}{\bibinfo{person}{Sijun Tan}, \bibinfo{person}{Brian Knott},
  \bibinfo{person}{Yuan Tian}, {and} \bibinfo{person}{David~J Wu}.}
  \bibinfo{year}{2021}\natexlab{}.
\newblock \showarticletitle{CryptGPU: Fast privacy-preserving machine learning
  on the GPU}. In \bibinfo{booktitle}{\emph{2021 IEEE Symposium on Security and
  Privacy (SP)}}. IEEE, \bibinfo{pages}{1021--1038}.
\newblock


\bibitem[Wagh(2022)]%
        {wagh2022pika}
\bibfield{author}{\bibinfo{person}{Sameer Wagh}.}
  \bibinfo{year}{2022}\natexlab{}.
\newblock \showarticletitle{Pika: Secure Computation using Function Secret
  Sharing over Rings}.
\newblock \bibinfo{journal}{\emph{Cryptology ePrint Archive}}
  (\bibinfo{year}{2022}).
\newblock


\bibitem[Wagh et~al\mbox{.}(2019)]%
        {wagh2019securenn}
\bibfield{author}{\bibinfo{person}{Sameer Wagh}, \bibinfo{person}{Divya Gupta},
  {and} \bibinfo{person}{Nishanth Chandran}.} \bibinfo{year}{2019}\natexlab{}.
\newblock \showarticletitle{SecureNN: 3-Party Secure Computation for Neural
  Network Training}.
\newblock \bibinfo{journal}{\emph{Proceedings on Privacy Enhancing
  Technologies}}  \bibinfo{volume}{2019} (\bibinfo{year}{2019}),
  \bibinfo{pages}{26 -- 49}.
\newblock


\bibitem[Wagh et~al\mbox{.}(2021)]%
        {wagh2021falcon}
\bibfield{author}{\bibinfo{person}{Sameer Wagh}, \bibinfo{person}{Shruti
  Tople}, \bibinfo{person}{Fabrice Benhamouda}, \bibinfo{person}{Eyal
  Kushilevitz}, \bibinfo{person}{Prateek Mittal}, {and} \bibinfo{person}{Tal
  Rabin}.} \bibinfo{year}{2021}\natexlab{}.
\newblock \showarticletitle{F: Honest-majority maliciously secure framework for
  private deep learning}.
\newblock \bibinfo{journal}{\emph{Proceedings on Privacy Enhancing
  Technologies}} \bibinfo{volume}{2021}, \bibinfo{number}{1}
  (\bibinfo{year}{2021}), \bibinfo{pages}{188--208}.
\newblock


\bibitem[Watson et~al\mbox{.}(2022)]%
        {watson2022piranha}
\bibfield{author}{\bibinfo{person}{Jean-Luc Watson}, \bibinfo{person}{Sameer
  Wagh}, {and} \bibinfo{person}{Raluca~Ada Popa}.}
  \bibinfo{year}{2022}\natexlab{}.
\newblock \showarticletitle{Piranha: A $\{$GPU$\}$ Platform for Secure
  Computation}. In \bibinfo{booktitle}{\emph{31st USENIX Security Symposium
  (USENIX Security 22)}}. \bibinfo{pages}{827--844}.
\newblock


\bibitem[Yao(1986)]%
        {yao1986generate}
\bibfield{author}{\bibinfo{person}{Andrew Chi-Chih Yao}.}
  \bibinfo{year}{1986}\natexlab{}.
\newblock \showarticletitle{How to generate and exchange secrets}.
\newblock \bibinfo{journal}{\emph{27th Annual Symposium on Foundations of
  Computer Science (sfcs 1986)}} (\bibinfo{year}{1986}),
  \bibinfo{pages}{162--167}.
\newblock


\end{thebibliography}










\end{document}